\documentclass[letterpaper, 10 pt, conference]{ieeeconf}  

\IEEEoverridecommandlockouts                              
\usepackage{tikz}
\usetikzlibrary{arrows,positioning,automata, fit,shapes}

\usepackage{multicol}
\usepackage{amsmath,amssymb,amsfonts}
\usepackage[font=small,labelfont=bf]{caption}
\usepackage{epsfig,graphics}
\usepackage{xcolor}
\usepackage[lined,ruled,linesnumbered]{algorithm2e}

\usepackage{pgf}
\usepackage{tikz}
\usetikzlibrary{arrows,automata}

\graphicspath{{./fig/}}

\usepackage{theorem}
\newtheorem{lemma}{Lemma}

\newtheorem{definition}{Definition}

\newtheorem{problem}{Problem}
\newtheorem{example}{Example}

\newcommand{\T}{\mathcal{T}} 
\newcommand{\TS}{\mathcal{T} = (S, \langle S_p,S_a \rangle, s_\init, U_p, U_a, T)} 
\newcommand{\A}{\mathcal{A}} 

\newcommand{\Plays}{\mathit{Plays}}

\newcommand{\init}{\mathit{init}}

\newcommand{\curr}{\mathit{curr}}

\newcommand{\Set}{\mathsf{S}} 

\newcommand{\unsafe}{\mathit{Unsafe}}
\newcommand{\safe}{\mathit{Safe}}
\newcommand{\losing}{\mathit{Losing}}

\newcommand{\cand}{{A_\mathit{cand}}}

\newcommand{\new}[1]{{#1}}

\title{\LARGE \bf
Synthesizing least-limiting guidelines for safety of \\ semi-autonomous systems
}

\author{Jana Tumova and Dimos V. Dimarogonas
\thanks{The authors are with the ACCESS Linnaeus Center, School of Electrical
Engineering, KTH Royal Institute of Technology, SE-100 44, Stockholm,
Sweden and with the KTH Centre for Autonomous Systems. This work was supported by the H2020 ERC Starting Grant BUCOPHSYS.
        {\tt\small \{tumova,dimos\}@kth.se}}%
}

\begin{document}

\maketitle
\thispagestyle{empty}
\pagestyle{empty}

\begin{abstract}
We consider the problem of synthesizing safe-by-design control strategies for semi-autonomous systems. Our aim is to address situations when safety cannot be guaranteed solely by the autonomous, controllable part of the system and a certain level of collaboration is needed from the uncontrollable part, such as the human operator. In this paper, we propose a systematic solution to generating least-limiting guidelines, i.e. the guidelines that restrict the human operator as little as possible in the worst-case long-term system executions. The algorithm leverages ideas from 2-player turn-based games.
\end{abstract}

\section{Introduction}

Recent technological developments have enhanced the application areas of autonomous and semi-autonomous cyber-physical systems to a variety of everyday scenarios from industrial automation to transportation and to housekeeping services. These examples have a common factor; they involve operation in an uncertain environment in the presence of highly unpredictable and uncontrollable agents, such as humans. In robot-aided manufacturing, there is a natural combination of autonomy and human contribution. In semi-autonomous driving, the vehicle is partially controlled automatically and partially by a human driver. Even in fully autonomous driving, passengers and pedestrians interact with the vehicle and actively influence the overall system safety and performance. The need for obtaining guarantees on behaviors of these systems is then even more crucial as the stakes are high. 
Formal verification and formal methods-based synthesis techniques were designed to provide such guarantees and recently, they have gained a considerable amount of popularity in applications to correct-by-design robot control. For instance, in~\cite{hadas09TL,nok-hscc2010} temporal logic control of robots in uncertain, reactive environments was addressed. In~\cite{marius-nondet} control synthesis for nondeterministic systems from temporal logic specifications was developed.
Loosely speaking, these works achieve the provable guarantees by accounting for the \emph{worst-case} scenarios in the control synthesis procedure. The uncertainty is therein treated as an \emph{adversary}, which however, often prevents the synthesis procedure to find a correct-by-design autonomous controller.

In this paper, we take a fresh perspective on correct-by-design control synthesis. We specifically focus on situations when the desired controller does not exist.
In contrast to the above mentioned approach, we view the uncertain, uncontrollable elements in the system as \emph{collaborative} in the sense that they have as much interest in keeping the overall system behavior safe, effective, and efficient as the autonomous controller does. At the same time, we still view them as to a large extent \emph{uncontrollable} in the sense that they still have their own intentions and we cannot force them to follow literal step-by-step instructions. In contrast, we aim to \emph{advise} them on what not to do if completely necessary, while keeping their options as rich as possible. 

\new{For example, consider a collaborative human-robot manufacturing task with the goal of assembling {products} $ABC$ through connecting pieces of types $A$ and $C$ to a piece of type $B$. The human operator can put together $A$ with $B$ or with $BC$, whereas the autonomous robot can put together $B$ or $AB$ with $C$. Our goal is to guarantee system safety, meaning that the human and the robot do not work with the same piece of type $B$ at the same time. While we can design a controller for the robot that does not reach for a piece being held by a human, we cannot guarantee that the human will not reach for a piece being held by the robot. To that end, we aim to synthesize \emph{guidelines} for the human, i.e. advise that reaching for a piece that the robot holds will lead to the safety violation. Under the assumption that the human follows this advise, the safety is guaranteed. Yet, this advise is still much less restrictive for the human operator than if the human-robot system was considered controllable as a whole. Namely, in such a case, a correct-by-design controller could dictate the human to always touch only solo $B$ pieces while the robot would be supposed to work only with $AB$ pieces pre-produced by the human. Clearly, the former mentioned guidelines allow for much more freedom of the human's decisions as the human may choose to work with an instance of $B$ piece or $BC$ piece.}
A similar situation occurs in an autonomous driving scenario with a pedestrian crossing the street. If the pedestrian jumps right in front of the car, the collision is unavoidable. A possible guideline for the human enabling the system safety would be not to ever cross the street. This is however a very limiting constraint. Instead, advising the human not to cross the street if the car is close seems quite reasonable.

This paper introduces a \emph{systematic way to synthesize least-limiting guidelines} for the uncontrollable elements in (semi-)autonomous systems, such as humans in human-robot systems, that allow the autonomous part of the system to maintain safety. Similarly as in some related work on correct-by-design control synthesis (e.g.,~\cite{marius-nondet}), we model the overall system state space as a two-player game on a graph with a safety winning condition. The autonomous, controllable entity takes the role of the game protagonist, whereas the uncontrollable entity is the adversary. We specifically work with situations, where the protagonist does not have a winning strategy in the game. We formalize the notion of \emph{adviser} as a function that ``forbids'' the application of certain adversary's inputs in certain system states. Furthermore, we classify the advisers based on the level of limitation they impose on the adversary. Finally, we provide an algorithm to find a least-limiting adviser that allows the protagonist to win the game, i.e. to keep the system safe. 
We also discuss the use of the synthesized advisers for on-the-fly guidance of the system execution. In this work, we do not focus on how the interface between the adviser and the uncontrollable element, such as human, should look like. Rather than that, the contribution of this paper can be summarized as the development of a theoretical framework for automated synthesis of reactive, least-limiting guidelines and control strategies that guarantee the system safety. 
\medskip

Related work includes literature on synthesis of environment assumptions that enable a winning game~\cite{environment} and on using counter-strategies for synthesizing assumptions in generalized reactivity (1) (GR(1)) fragment of LTL~\cite{mining,counter}. These works however synthesize the assumptions in the form of logic formulas, whereas we focus on guiding the adversary through explicitly enumerating the inputs that should not be applied. Synthesis of maximally permissive strategies is considered in~\cite{permissive} and also in discrete-event systems literature in~\cite{mp-sofsem}, where however, only controllable inputs are being restricted.
Our approach is different to the above works, since we aim for systematic construction of reactive guidelines in the sense that if the least-limiting adviser is not followed, a suitable substitute adviser is supplied if such exists. We also use a different criterion to measure the level of limitation that is the worst-case long-term average of restrictions as opposed to the cumulative number of restrictions considered in~\cite{environment} or the size of the set of behaviors considered in~\cite{permissive}.
Other related literature studies problems of minimal model repair~\cite{model-repair,mr-mdp}, synthesis of least-violating strategies~\cite{faella,hscc}, or design of reward structures for
decision-making processes in context of human-machine interaction~\cite{mazo}. This work can be also viewed in the context of literature aimed at collaborative human-robot control, e.g.,~\cite{ketan,sandra}. 

\medskip

The paper is structured as follows. In Sec.~\ref{sec:prelims} we introduce necessary notation and preliminaries. In Sec.~\ref{sec:problem}, we state our problem. In Sec.~\ref{sec:solution}, we introduce the synthesis algorithm in details and discuss the use of the synthesized solution for on-the-fly guidance. Sec.~\ref{sec:conc} concludes the paper and discusses future research. Throughout the paper, we provide several illustrative examples demonstrating the developed theory.

%
%
%

\section{Preliminaries}
\label{sec:prelims}


\label{sec:prelims}

Given a set $\Set$, we use $2^\Set$, $|\Set|$, $\Set^*$, $\Set^\omega$ to denote the powerset of $\Set$, the cardinality of $\Set$, and the set of all finite and infinite sequences of elements from $\Set$, respectively. Given a finite sequence $w$ and a finite or an infinite sequence $w'$, we use $w \cdot w'$ to denote their concatenation. \new{Let $w(i)$ and $w_{\leadsto j}$ denote the $i$-th element of word $w$ and the prefix of $w$ that ends in $w(j)$, respectively}. Furthermore, assuming that $\Set$ is a set of finite sequences and $\Set'$ is a set of finite and/or infinite sequences, $\Set \cdot \Set' = \{w \cdot w' \mid w \in \Set \wedge \, w' \in \Set'\}$. $\mathbb{Z}$ denotes the set of integers.

\label{sec:games}
\begin{definition}[Arena]\label{def:arena}
A 2-player turn-based {game arena} is a transition system $\TS$, where
 $S$ is a \new{nonempty, finite} set of states;
$\langle S_p, S_a \rangle$ is a partition of $S$ into the set of protagonist (player $p$) states $S_p$ and the set of adversary (player $a$) states $S_a$, such that $S_p \cap S_a = \emptyset$, $S_p \cup S_a = S$;
$s_\init \in S_p$ is the initial protagonist state;
$U_p$ is the set of inputs of the protagonist;
$U_a$ is the set of inputs of the adversary;
$T = T_p \cup T_a$, is a partial injective transition function, where $T_p: S_p \times U_p  \rightarrow S_a$ and $T_a: S_a \times U_a \rightarrow S_p$. 
\end{definition}

Note that in a protagonist state, only an input of the protagonist can be applied, and analogously, in an adversary state, only an input of the adversary can be applied. We assume that from a protagonist state, the system can only transition to an adversary state and vice versa. This assumption is not restrictive, since it can be easily shown that any game arena with $T_p: S_p \times U_p \rightarrow S$ and $T_a: S_a \times U_a \rightarrow S$ can be transformed to satisfy it. Loosely speaking, each transition from a protagonist state to a protagonist state is split into two transitions, to and from a  new adversary state. Analogous transformation can be applied to the transitions from adversary states to adversary states. 

Let $U_i^{s_i} = \{u_i \in U_i \mid T_i(s_i,u_i) \text{ is defined} \}$ denote the set of inputs of player $i \in \{p,a\}$ that are \emph{enabled} in the state $s_i \in S_i$. Arena $\T$ is \emph{non-blocking} if $|U_i^{s_i}| \geq 1$, for all $i \in \{p,a\}$ and all $s_i \in S_i$ \new{and \emph{blocking} otherwise.}
A \emph{play} in $\T$ is an \emph{infinite} alternating sequence of protagonist and adversary states $\pi = s_{p,1}s_{a,1}s_{p,2}s_{a,2} \ldots$, 
such that $s_{p,1} = s_\init$ and for all $j \geq 1$ there exist $u_{p,j} \in U_p, u_{a,j} \in U_a$, such that $T_p(s_{p,j},u_{p,j}) = s_{a,j}$, and $T_a(s_{a,j},u_{a,j}) = s_{p,j+1}$. 
Note that for each play $\pi$, $\pi(2k) \in S_a$, while $\pi(2k-1) \in S_p$, for all $1 \leq k$. \new{ A \emph{play prefix} $\pi_{\leadsto j}=\pi(1)\ldots \pi(j)$ is a finite prefix of a play $\pi =\pi(1)\pi(2) \ldots$. Let $\Plays^\T$ denote the set of all plays in $\T$.}
\new{If a set of plays $\mathit{Plays}^{\dot \T}$ of a blocking arena $\dot{\T}$ is nonempty, then $\dot \T$ can be transformed into an equivalent non-blocking arena $\T$ via a systematic removal of \emph{blocking states} and their adjacent transitions that are defined inductively as follows: (i) each $s_i \in S_i$, $i \in \{p,a\}$, such that $U_i^{s_i} = \emptyset$  is a blocking state and (ii)~if $T_i(s_i,u_i)$ is a blocking state for each $u_i \in U_i^{s_i}$, then $s_i$, $i \in \{i,p\}$ is a blocking state, too. Then $\Plays^{\dot \T} = \Plays^{\T}$. }


A \emph{deterministic control strategy} (or strategy, for short) of player $i \in \{p,a\}$ is a partial function $\sigma_i^\T: S^*\cdot S_i \rightarrow U_i$ that assigns a player $i$'s enabled input $u_i \in U_i^{s_i}$ to \new{each play prefix} in $\T$ that ends in a player $i$'s state $s_i \in S_i$.
Strategies $\sigma_p^\T,\sigma_a^\T$ \emph{induce a play} $\pi^{\sigma_p^\T,\sigma_a^\T} = s_{p,1} s_{a,1} s_{p,2} s_{a,2} \ldots \in (S_p \cdot S_a)^\omega$, such that $s_{p,1} = s_\init$, and for all $j \geq 1$, $T_p(s_{p,j},\sigma_{p}(s_{p,1} s_{a,1} \ldots s_{p,j})) = s_{a,j}$, and $T_a(s_{a,j},\sigma_{a}(s_{p,1} s_{a,1} \ldots s_{p,j}s_{a,j})) = s_{p,j+1}$. 
A  strategy $\sigma_i^\T$ is called \emph{memoryless} if it satisfies the property that $\sigma_i^\T(s_1 \ldots s_n) = \sigma_i^\T(s_1' \ldots s_m')$ whenever $s_n = s_m'$. Hence, with a slight abuse of notation, memoryless control strategies are viewed as functions $\varsigma_i^\T: S_i \rightarrow U_i$. 
The set of all strategies of player $i$ in  $\T$ is denoted by $\Sigma_i^\T$. The set of all plays induced by all strategies in $\Sigma_p^\T,\Sigma_a^\T$, i.e. the set of all plays in $\T$ is $\Plays^{\Sigma_p^\T,\Sigma_a^\T} = \{\pi^{\sigma_p^\T,\sigma_a^\T} \mid \sigma_p^\T \in \Sigma_p^\T, \sigma_a^\T \in \Sigma_a^\T\}.$ Analogously, we use $\Plays^{\sigma_p^\T,\Sigma_a^\T} = \{\pi^{\sigma_p^\T,\sigma_a^\T} \mid \sigma_a^\T \in \Sigma_a^\T\}$ to denote the set of plays induced by a given strategy $\sigma_p^\T$ and by all strategies $\sigma_a^\T \in \Sigma_a^\T$.

{A \emph{game} $G = (\T,W)$ consists of a game arena $\T$ and a \emph{winning condition} $W \subseteq Plays^{\Sigma_p^\T,\Sigma_a^\T}$ that is in general a subset of plays in $\T$. 
A \emph{safety winning condition} is $W_\safe = \{\pi \in \Plays^{\Sigma_p^\T, \Sigma_a^\T} \mid \text{ for all } j \geq 1. \, \pi(j) \in \safe\},$ 
where $S = \langle \safe, \unsafe \rangle$ is a partition of the set of states into the safe and unsafe state subsets.}
A protagonist's strategy $\sigma_p^\T$ is winning if $\Plays^{\sigma_p^\T,\Sigma_a^\T} \subseteq W$. Let $\Omega_p^\T \subseteq \Sigma_p^\T$ denote the set of all protagonist's winning strategies.



\medskip

Let $\TS$ be an arena and $w: S  \times S \rightarrow \mathbb{Z}$ be a weight function that assigns a weight to each $(s,s')$, such that there exists $u \in U_p \cup U_a$, where $(s,u,s') \in T$. Then $(\T,w)$ can be viewed as an arena of a \emph{mean-payoff game}. 
The value secured by protagonist's strategy $\sigma_p^\T$ is 
\vspace{-0.2cm}
$$\nu(\sigma_p^\T) = \inf_{\sigma_a^\T \in \Sigma_a^\T} \liminf_{n \rightarrow \infty} \frac{1}{n} \sum_{j=1}^{n} w(\pi^{\sigma_p^\T,\sigma_a^\T}(j),\pi^{\sigma_p^\T,\sigma_a^\T}(j+1) ).$$
\vspace{-0.3cm}

An \emph{optimal protagonist's strategy} $\sigma_p^{\T*}$ secures the optimal value $\nu(\sigma_p^{\T*}) = \sup_{\sigma_p^\T \in \Sigma_p^\T} \nu(\sigma_p^\T).$
Several algorithms exist to find the optimal protagonist's strategy, see, e.g., \cite{brim}.
For more details on games on graphs in general, we refer the interested reader e.g., to~\cite{games-book}.

\section{Problem formulation}
\label{sec:problem}



The \emph{system} that we consider consists of two entities: the first one is the autonomous part of the system that we aim to control (e.g., a robotic arm), and the second one is the agent that is uncontrollable, and to a large extent unpredictable (e.g., a human operator in a human-robot manufacturing scenario). The overall state of such system is determined by the system states of these entities (e.g., the positions of the robotic and the human arms and objects in their common workspace and the status of the manufacturing). In this paper, we consider systems with a finite number of states $Q$ (obtained, e.g., by partitioning the workspace into cells). The system state can change if one of the entities takes a decision and applies an input (e.g., the robot can move the arm from on cell to another, or the human can pick up an object). For simplicity, we assume that the entities take regular turns in applying their inputs. This assumption is however not too restrictive as we may allow the entities to apply a special pass input $\epsilon$ that does not induce any change to the current system state.


\medskip


To model the system formally, we call the former, controllable entity the protagonist, the latter, uncontrollable entity the adversary, and we capture the impacts of their inputs to the system states through a game arena (see Def.~\ref{def:arena})
\begin{align}
  \TS.
\end{align}
The set of the arena states is $S = Q \times \{p,a\}$ and each arena state $s = (q,i) \in S$ is defined by the system state $q \in Q$ and the entity $i \in \{p,a\}$ whose turn it is to apply its input, i.e. $(q,p) \in S_p$, and $(q,a) \in S_a$, for all $q \in Q$. Behaviors of the system are thus captured through plays in the arena.




The goal of the former, controllable entity is to keep the system safe, i.e. to avoid the subset of unsafe system states, while the latter entity has its own goals, such as to reach a certain system state, etc. Formally, the protagonist is given a partition of states $S = \langle \safe,\unsafe \rangle$ and the corresponding safety winning condition $W_\safe$. The arena $\T$ together with the safety winning condition $W_\safe$ establish a game $(T,W_\safe)$.

\vspace{-0.3cm}

{\begin{example}

Consider the simplified manufacturing scenario outlined in the introduction.  A system state is determined by the current pieces in the workspace and their status; each of them is either on the desk, held by the human, or by the robot: $Q \subseteq 2^{\{A,B,C,AB,BC,ABC\} \times \{\mathit{desk,human,robot}\}}$. The robot acts as the protagonist and the human as the adversary. $s_ \init = \big(\{\mathit{(A,desk), (B,desk),(C,desk)\},a} \big)$ is an example of a system initial state. The inputs of the robot are $U_p = \big\{\{\mathit{grab_p, drop_p}\} \times \{A,B,C,AB,BC,ABC\} \cup \{\mathit{connect_p}\} \times \{(B,C),(AB,C)\}\big\}$ and similarly, $U_a = \big\{\mathit{grab_a, drop_a}\}  \times \{A,B,C,AB,BC,ABC\} \cup \{\mathit{connect_a}\} \times \{(A,B),(A,BC)\}\} $. The transition function reflects the effect of inputs on the system state. For instance, 
\vspace{-0.5cm}

\small{
\begin{align*}
T&\big(\big(
\{(A,desk), (B,desk),(C,desk)\},a\big),(grab_a,A)\big) = \\ =  & \ \big(\{(A,human), (B,desk),(C,desk)\},p\big), \text{ or } \\ 
T&\big(\big(\{\mathit{(A,desk), (B,robot),(C,robot)\},p\big),\big(connect_p,(B,C)\big)\big)} \\ = & \  \big(\{(A,human),(AB,robot)\},a\big)\big).
\end{align*}}

\vspace{-0.5cm}
\normalsize
Note that the transition function does not have to be manually enumerated. Rather than that, it can be generated from conditions, such as
$T \big(\big(
\{(x,y)\}\cup Z,a\big),(grab_a,x)\big) = \big(\{(x,human)\} \cup Z,p\big)$,
applied to all $x \in \{A,B,C,AB,AC,ABC\}, y \in \{desk,robot\}, Z \subseteq (\{A,B,C,AB,AC,ABC\}\setminus \{x\}) \times \{desk,human,robot
\}$.

\end{example}}

\label{sec:goal}

The problem of finding a protagonist's winning control strategy $\sigma_p^\T$ guaranteeing system safety has been studied before and even more complex winning conditions have been considered \cite{games-book}. 
In this work, we focus on a situation when the protagonist \emph{does not} have a winning control strategy. For such cases, we aim to generate a least-limiting subset of adversary's control strategies that would permit the protagonist to win.  
Loosely speaking, this subset can be viewed as the minimal guidelines for the adversary's collaboration. 

{Note that this problem differs from the supervisory control of discrete event systems as we do not limit only the application of controllable, but also the uncontrollable inputs. However, it also differs from the synthesis of controllers for fully controllable systems as we aim to limit the adversary's application of uncontrollable inputs as little as possible.}
We formalize the guidelines for the adversary's collaboration through the notion of adviser and adviser restricted arena.

\vspace{-0.2cm}

\begin{definition}[Adviser]
An adviser is a mapping $\alpha: S_a \rightarrow 2^{U_a}$, where $\alpha(s_a) \subseteq U_a^{s_a}$ represents the subset of adversary's inputs that are forbidden in state $s_a$.

Given an arena $\T = (S, \langle S_p,S_a \rangle, s_\init, U_p, U_a, T_p \cup T_a)$, and an adviser $\alpha$, the adviser restricted arena is $\dot{\T}^{\alpha} = (S, \langle S_p,S_a \rangle, s_\init, U_p, U_a, \dot{T}_p^{\alpha} \cup \dot{T}_a^{\alpha}),$ where $\dot{T}_p^{\alpha} = T_p$ and $\dot{T}_a^{\alpha} = T_a \setminus \{(s_a,u_a,s_p) \mid u_a \in \alpha(s_a)\}$. The set of all plays in $\dot{\T}^{\alpha}$ is denoted by $\Plays^{\dot\alpha}$.
\label{def:restricted}
\end{definition}

\new{If $\alpha(s_a) = U_a^{s_a}$ for some $s_a \in S_a$, the adviser restricted arena becomes blocking, and hence, not every sequence $s_{p,1}s_{a,1}s_{p,2}s_{a,2} \ldots s_{a,k}$, satisfying $s_{p,1} = s_\init$,  and for all $1 \leq j \leq k$, $1 \leq \ell < k$, $\dot{T}_p^\alpha(s_{p,j},\sigma_{p}(s_{p,1} s_{a,1} \ldots s_{p,j})) = s_{a,j}$, and $\dot{T}_a^\alpha(s_{a,\ell},\sigma_{a}(s_{p,1} s_{a,1} \ldots s_{p,\ell}s_{a,\ell})) = s_{p,\ell+1}$, can be extended to a play. However, if $\Plays^{\dot\alpha}$ is nonempty,
we can transform $\dot \T^\alpha$ into a \emph{non-blocking adviser restricted arena} 
\begin{align}\T^\alpha = (S^
\alpha, \langle S_p^\alpha,S_a^\alpha \rangle, s_\init, U_p, U_a, T_p^\alpha \cup T_a^\alpha)\end{align} that has the exact same set of plays $\Plays^\alpha = \Plays^{\dot\alpha}$ as $\dot \T^\alpha$ as outlined in Sec.~\ref{sec:games}. 
Let us denote the sets of all protagonist's and adversary's strategies in $\T^\alpha$ by $\Sigma_p^\alpha$ and $\Sigma_a^\alpha$, respectively. $\Plays^{\sigma_p^\alpha,\Sigma_a^\alpha}$ refers to the set of plays induced by $\sigma_p^\alpha \in \Sigma_p^\alpha$ and $\Sigma_a^\alpha$ in $\T^\alpha$. 
If however $\dot{P}\mathit{lays}^{\alpha}$ is empty, a non-blocking adviser restricted arena $\T^\alpha$ does not exist.

Given the winning condition $W_\safe$, we define a good adviser $\alpha$ as one that permits the protagonist to achieve safety in the non-blocking adviser restricted arena $\T^\alpha$.
}





\vspace{-0.2cm}

\begin{definition}[Good adviser]
An adviser $\alpha$ is \emph{good} for $(\T,W_\safe)$ if there exists a non-blocking adviser restricted arena $\T^\alpha$ and a protagonist's strategy $\sigma_p^\alpha \in \Sigma_p^\alpha$, such that $\Plays^{\sigma_p^\alpha,\Sigma_a^\alpha} \subseteq W_\safe$. Given a good adviser $\alpha$, the set of protagonist's winning strategies is denoted by $\Omega_p^\alpha \subseteq \Sigma_p^\alpha$.

\end{definition}

\vspace{-0.2cm}
Since there might be more good advisers, we need to distinguish which of them limit the adversary less and which of them more. To that end, we associate each adviser with a cost, called adviser level of limitation. 

\vspace{-0.2cm}

\begin{definition}[Adviser level of limitation]
Given  an arena $\T$ and a good adviser $\alpha$, we define the \emph{adviser level of limitation}
\begin{align}
 \lambda(\alpha)& = \inf_{\sigma_p^\alpha \in \Omega_p^\alpha} \gamma(\sigma_p^\alpha),\text{ where} \label{eq:alpha}\\
\gamma(\sigma_p^\alpha) & = \sup_{\sigma_a^\alpha \in \Sigma_a^\alpha} \limsup_{n \rightarrow \infty} \frac{1}{n} \sum_{j=1}^{n} \big|\alpha(\pi^{\sigma_p^\alpha,\sigma_a^\alpha}(2j))\big|.
  \label{eq:gamma}
\end{align}
\end{definition}

In other words, $\lambda(\alpha)$ is the  \emph{worst-case long-term average} of the number of forbidden inputs along the plays induced by the \emph{best-case} protagonist's strategy $\sigma_p^\alpha$. The choice of the worst-case long-term average is motivated by the fact that although the adversary can be advised, it cannot be controlled. On the other hand, the consideration of the best-case $\sigma_p^\alpha$ is due to the protagonist being fully controllable.
We provide some intuitive explanations on the introduced terminology through the following illustrative example.

\vspace{-0.2cm}
\begin{example}[Safety game and adviser]
\new{
An example of a game arena with a safety winning condition $W_\safe$
is given in Fig.~\ref{fig:example}. (A). The squares illustrate the protagonist's states and the circles illustrate the adversary's ones. Transitions are depicted as arrows between them and they are labeled with the respective inputs that trigger them. The safe states in $\safe$ are shown in green and the unsafe ones in $\unsafe$ are in blue.
Fig.~\ref{fig:example}.(B)-(D) show three advisers $\alpha_B,\alpha_C$ and $\alpha_D$, respectively, via marking the forbidden transitions in red. In Fig.~\ref{fig:example}.(B), $\alpha_B(s_2)$ = $\{u_{a_3}\}$, $\alpha_B(s_4) = \{u_{a_4},u_{a_5}\}$, and $\alpha_B(s_6) = \{u_{a_6},u_{a_7}\}$. In Fig.~\ref{fig:example}.(C), $\alpha_C(s_2) = \{u_{a_3}\}$ and $\alpha_C(s_4) = \alpha_C(s_6) = \emptyset$. Finally, in  Fig.~\ref{fig:example}.(D), $\alpha_D(s_2) = \{u_{a_2},u_{a_3}\}$ and
$\alpha_D(s_4) = \alpha_D(s_6) = \emptyset$.

\begin{figure}[!t]
\vspace{0.2cm}
(A)

\small
\begin{center}
\usetikzlibrary{arrows,positioning,automata,shadows,fit,shapes}
\begin{tikzpicture}[->,>=stealth',shorten >=1pt,auto, node distance=1.8cm,semithick,initial text=]

  \tikzstyle{every state}=[draw]
  \tikzstyle{a}=[circle,thick,draw=green!99,fill=green!30,minimum size=10mm]
  \tikzstyle{aunsafe}=[a,draw=blue!75,fill=blue!20]
  \tikzstyle{p}=[rectangle,thick,draw=green!99,
  			  fill=green!30,minimum size=6mm]
  \tikzstyle{punsafe}=[p,draw=blue!75,fill=blue!20]
  			  
	  \node[initial, initial where=above, p, minimum size=4ex] (q1) {$s_1$}; 
	  \node[a, minimum size=4.5ex] (q2) [below left = of  q1] {$s_2$}; 
      \node[p, minimum size=4ex] (q21) [right of = q2] {$s_3$}; 
	  \node[a, minimum size=4.5ex] (q3) [below right = of q1] {$s_4$}; 
	  \node[punsafe, minimum size=4ex] (q4) [below left = of q3] {$s_5$};
      \node[aunsafe, minimum size=4.5ex] (q41) [right of = q4] {$s_6$};  
	  \node[punsafe, minimum size=4ex] (q5) [below right = of  q3] {$s_7$}; 
	
\begin{scope}[every node/.style={scale=1}]
  \path (q1) edge node [pos=0.1] {$u_{p_1}$} (q2);
    \path (q1) edge node [pos=0.2,right] {$u_{p_2}$} (q3);
  \path (q2) edge [bend left] node [pos=0.2,left] {$u_{a_1}$} (q1);
  \path (q2) edge [bend right] node [pos=0.2,below] {$u_{a_2}$} (q21);
  \path (q2) edge [bend right] node [pos=0.2,left] {$u_{a_3}$} (q4);
  \path (q21) edge node [pos=0.3,above ] {$u_{p_3}$} (q2);
  \path (q21) edge node {$u_{p_4}$} (q3);
  \path (q3) edge [bend right] node  [pos=0.2] {$u_{a_4}$} (q4);
  \path (q3) edge [bend left] node  [pos=0.15, right] {$u_{a_5}$} (q5);
  \path (q4) edge node [pos=0.3] {$u_{p_5}$} (q41);
  \path (q5) edge node [pos=0.5,above] {$u_{p_6}$} (q41);
  \path (q5) edge node [pos=0.15,left]{$u_{p_7}$} (q3);
  \path (q41) edge [bend left] node  [pos=0.2,below] {$u_{a_6}$} (q4);
  \path (q41) edge [bend right] node  [pos=0.2, below] {$u_{a_7}$} (q5);
 \end{scope}
\end{tikzpicture} \\
\end{center}
\normalsize{(B) \hspace{0.25\linewidth} (C) \hspace{0.25\linewidth} (D)}
\vspace{-0.6cm}
\small
\begin{center}
\scalebox{0.55}{
\usetikzlibrary{arrows,positioning,automata,shadows,fit,shapes}
\begin{tikzpicture}[->,>=stealth',shorten >=1pt,auto, node distance=1.2cm,semithick,initial text=]

  \tikzstyle{every state}=[draw]
  \tikzstyle{a}=[circle,thick,draw=green!99,fill=green!30,minimum size=10mm]
  \tikzstyle{aunsafe}=[a,draw=blue!75,fill=blue!20]
  \tikzstyle{p}=[rectangle,thick,draw=green!99,
  			  fill=green!30,minimum size=6mm]
  \tikzstyle{punsafe}=[p,draw=blue!75,fill=blue!20]

	  \node[initial, initial where=above, p, minimum size=4ex] (q1) {}; 
	  \node[a, minimum size=4.5ex] (q2) [below left = of  q1] {}; 
      \node[p, minimum size=4ex] (q21) [right of = q2] {}; 
	  \node[a, minimum size=4ex] (q3) [below right = of q1] {}; 
	  \node[punsafe, minimum size=4ex] (q4) [below left = of q3] {};
      \node[aunsafe, minimum size=4.5ex] (q41) [right of = q4] {};  
	  \node[punsafe, minimum size=4ex] (q5) [below right = of  q3] {}; 
	
\begin{scope}[every node/.style={scale=1}]
  \path (q1) edge node [pos=0.1] {} (q2);
    \path (q1) edge node [pos=0.2,right] {} (q3);
  \path (q2) edge [bend left] node [pos=0.2,left] {} (q1);
  \path (q2) edge [bend right] node [pos=0.2,below] {} (q21);
  \path (q2) edge [bend right,red!75,thick] node [pos=0.2,left,red!75] {} (q4);
  \path (q21) edge node [pos=0.3,above ] {} (q2);
  \path (q21) edge node {} (q3);
  \path (q3) edge [bend right,red!75,thick] node  [pos=0.2,red!75] {} (q4);
  \path (q3) edge [bend left,red!75,thick] node  [pos=0.15, right,red!75] {} (q5);
  \path (q4) edge node [pos=0.3] {} (q41);
  \path (q5) edge node [pos=0.5,above] {} (q41);
  \path (q5) edge node [pos=0.15,left]{} (q3);
  \path (q41) edge [bend left,red!75,thick] node  [pos=0.2,below] {} (q4);
  \path (q41) edge [bend right,red!75,thick] node  [pos=0.2, below] {} (q5);
 \end{scope}
\end{tikzpicture}
\begin{tikzpicture}[->,>=stealth',shorten >=1pt,auto, node distance=1.2cm,semithick,initial text=]

  \tikzstyle{every state}=[draw]
  \tikzstyle{a}=[circle,thick,draw=green!99,fill=green!30,minimum size=10mm]
  \tikzstyle{aunsafe}=[a,draw=blue!75,fill=blue!20]
  \tikzstyle{p}=[rectangle,thick,draw=green!99,
  			  fill=green!30,minimum size=6mm]
  \tikzstyle{punsafe}=[p,draw=blue!75,fill=blue!20]

	  \node[initial, initial where=above, p, minimum size=4ex] (q1) {}; 
	  \node[a, minimum size=4.5ex] (q2) [below left = of  q1] {}; 
      \node[p, minimum size=4ex] (q21) [right of = q2] {}; 
	  \node[a, minimum size=4ex] (q3) [below right = of q1] {}; 
	  \node[punsafe, minimum size=4ex] (q4) [below left = of q3] {};
      \node[aunsafe, minimum size=4.5ex] (q41) [right of = q4] {};  
	  \node[punsafe, minimum size=4ex] (q5) [below right = of  q3] {}; 
	
\begin{scope}[every node/.style={scale=1}]
  \path (q1) edge node [pos=0.1] {} (q2);
    \path (q1) edge node [pos=0.2,right] {} (q3);
  \path (q2) edge [bend left] node [pos=0.2,left] {} (q1);
  \path (q2) edge [bend right] node [pos=0.2,below] {} (q21);
  \path (q2) edge [thick, bend right,red!75] node [pos=0.2,left] {} (q4);
  \path (q21) edge node [pos=0.3,above ] {} (q2);
  \path (q21) edge node {} (q3);
  \path (q3) edge [bend right] node  [pos=0.2] {} (q4);
  \path (q3) edge [bend left] node  [pos=0.15, right] {} (q5);
  \path (q4) edge node [pos=0.3] {} (q41);
  \path (q5) edge node [pos=0.5,above] {} (q41);
  \path (q5) edge node [pos=0.15,left]{} (q3);
  \path (q41) edge [bend left] node  [pos=0.2,below] {} (q4);
  \path (q41) edge [bend right] node  [pos=0.2, below] {} (q5);
 \end{scope}
\end{tikzpicture}
\begin{tikzpicture}[->,>=stealth',shorten >=1pt,auto, node distance=1.2cm,semithick,initial text=]

  \tikzstyle{every state}=[draw]
\tikzstyle{a}=[circle,thick,draw=green!99,fill=green!30,minimum size=10mm]
  \tikzstyle{aunsafe}=[a,draw=blue!75,fill=blue!20]
  \tikzstyle{p}=[rectangle,thick,draw=green!99,
  			  fill=green!30,minimum size=6mm]
  \tikzstyle{punsafe}=[p,draw=blue!75,fill=blue!20]

	  \node[initial, initial where=above, p, minimum size=4ex] (q1) {}; 
	  \node[a, minimum size=4.5ex] (q2) [below left = of  q1] {}; 
      \node[p, minimum size=4ex] (q21) [right of = q2] {}; 
	  \node[a, minimum size=4ex] (q3) [below right = of q1] {}; 
	  \node[punsafe, minimum size=4ex] (q4) [below left = of q3] {};
      \node[aunsafe, minimum size=4.5ex] (q41) [right of = q4] {};  
	  \node[punsafe, minimum size=4ex] (q5) [below right = of  q3] {}; 
	
\begin{scope}[every node/.style={scale=1}]
  \path (q1) edge node [pos=0.1] {} (q2);
    \path (q1) edge node [pos=0.2,right] {} (q3);
  \path (q2) edge [bend left] node [pos=0.2,left] {} (q1);
  \path (q2) edge [thick,bend right,red!75] node [pos=0.2,below] {} (q21);
  \path (q2) edge [thick,bend right,red!75] node [pos=0.2,left] {} (q4);
  \path (q21) edge node [pos=0.3,above ] {} (q2);
  \path (q21) edge node {} (q3);
  \path (q3) edge [bend right] node  [pos=0.2] {} (q4);
  \path (q3) edge [bend left] node  [pos=0.15, right] {} (q5);
  \path (q4) edge node [pos=0.3] {} (q41);
  \path (q5) edge node [pos=0.5,above] {} (q41);
  \path (q5) edge node [pos=0.15,left]{} (q3);
  \path (q41) edge [bend left] node  [pos=0.2,below] {} (q4);
  \path (q41) edge [bend right] node  [pos=0.2, below] {} (q5);
 \end{scope}
\end{tikzpicture}
}
\end{center}
\vspace{-0.2cm}

\caption{(A) An example of a game arena with a safety winning condition. The protagonist's and adversary's states are illustrated as squares and circles, respectively. The safe set $\safe$ is in green, the unsafe set $\unsafe$ in blue. Transitions are depicted as arrows between them and they are labeled with the respective inputs that trigger them. (B) -- (D) show three different advisers $\alpha_B,\alpha_C$ and $\alpha_D$, respectively, via marking the forbidden transitions in red.}
\label{fig:example}
\vspace{-0.7cm}
\end{figure}

For $\alpha_B$, the non-blocking adviser restricted arena contains states $S^{\alpha_B} = \{s_1,s_2,s_3\}$. The set of protagonist's strategies in $\T^{\alpha_B}$ is $\Sigma_p^{\alpha_B}=\{\sigma_p^{\alpha_B}\}$, such that $\sigma_p^{\alpha_B}(\pi(1) \ldots \pi(2j) s_1) = u_{p_1}$ and $\sigma_p^{\alpha_B}(\pi(1) \ldots \pi(2j) s_3) = u_{p_3}$, for all play prefixes $\pi(1) \ldots \pi(2j)$, $j\geq 0$ of all plays $\pi \in \Plays^{\Sigma_p^{\alpha_B},\Sigma_a^{\alpha_B}}$. Since $\sigma_p^{\alpha_B}$ is winning, $\alpha_B$ is good. 
It is easy to see that the set of protagonist's winning strategies  and the set of all adversary's strategies in $\T^{\alpha_B}$ induce a set of plays $\Plays^{\Omega_{p}^{\alpha_B},\Sigma_{a}^{\alpha_B}} =  \{ s_1s_2\pi(3)s_2\pi(5)s_2\pi(7)s_2\ldots \mid \pi(2j+1) \in \{s_1,s_3\}, \text{ for all } j\geq 1\}.$
The strategy $\sigma_p^{\alpha_B} \in \Omega_p^{\alpha_B}$ is therefore associated with the value
$ \gamma(\sigma_p^{\alpha_B})  =  \sup_{\sigma_a^{\alpha_B} \in \Sigma_{a}^{\alpha_B}}\limsup_{n \rightarrow \infty} \frac{1}{n} \sum_{j=1}^{n} \big|\alpha_B(\pi^{\sigma_p^{\alpha_B},\sigma_a^{\alpha_B}}(2j))\big| =  \limsup_{n \rightarrow \infty} \frac{1}{n} \sum_{j=1}^{n} \big|\alpha_B(s_2)\big| = 1,$ and the level of limitation of $\alpha_B$ is $\lambda(\alpha_B) = 1$.
Although it might seem that adviser $\alpha_B$ is more limiting than $\alpha_C$, it is not the case. The non-blocking adviser restricted arena $\T^{\alpha_C}$ in this case contains all states from $\T$, $S^{\alpha_C} = S$. However, the set of winning protagonist's strategies $\Omega_p^{\alpha_C}$ in $\T^{\alpha_C}$ is analogous as in case (B). Namely, if $\sigma_p^{\alpha_C}(\pi(1) \ldots \pi(2j) s_1) = u_{p_2}$ or $\sigma_p^{\alpha_C}(\pi(1)\pi(2) \ldots \pi(2j) s_3) = u_{p_4}$, the resulting play would not be winning for the protagonist as all adversary's choices in $s_4$ lead to an unsafe state. Hence, 
$\Plays^{\Omega_{p}^{\alpha_C},\Sigma_{a}^{\alpha_C}} =  \Plays^{\Omega_{p}^{\alpha_B},\Sigma_{a}^{\alpha_B}}$
and the level of limitation of $\alpha_C$ is $\lambda(\alpha_C) = 1$.
Finally, $\alpha_D$ is more limiting than $\alpha_B$ and $\alpha_C$. Following similar reasoning as above, we can see that  $\Plays^{\Omega_{p}^{\alpha_D},\Sigma_{a}^{\alpha_D}} = \{s_1s_2s_1s_2s_1s_2s_1s_2\ldots\},$ but since $|\alpha_D(s_2)|=2$, we have  
$\lambda(\alpha_D)  = \inf_{\sigma_p^{\alpha_D} \in \Omega_p^{\alpha_D}} \gamma(\sigma_p^{\alpha_D})  =   \limsup_{n \rightarrow \infty} \frac{1}{n} \sum_{j=1}^{n} \big|\alpha_D(s_2)\big| = 2.$}

\end{example}




\new{
\begin{problem}
Consider $\TS$, and a safety winning condition $W_\safe$ given via a partition $S = \langle \safe, \unsafe \rangle$. Synthesize an adviser $\alpha^\star$, and a {protagonist's winning strategy} $\sigma_p^{\alpha^\star\star}$, such that:
\begin{itemize}
\item[(i)] $\alpha^\star$ is good and
$\sigma_p^{\alpha^\star\star} \in \Omega_p^{\alpha\star}$, 
\item[(ii)] $\lambda (\alpha^\star) = \inf_{\alpha \in A} \lambda(\alpha)$, where $A$ is the set of all good advisers for $(\T,W_\safe)$, i.e. $\lambda (\alpha^\star) $ is least-limiting and 
\item[(iii)] $\gamma(\sigma_p^{\alpha^\star\star}) = \inf_{\sigma_p^{\alpha^\star} \in \Omega_p^{\alpha^\star}}\gamma(\sigma_p^{\alpha^\star})$, i.e. $\sigma_p^{\alpha^\star\star}$ is optimal.
\end{itemize}
\label{problem:main}
\end{problem}
}

\section{Solution}

\label{sec:solution}

Our solution builds on several steps: first, we generate a so-called \emph{nominal adviser}, which assigns to each adversary state the set of forbidden inputs. We prove that the nominal adviser is by construction good, but does not have to be least-limiting. Second, building on the nominal adviser, we efficiently generate a finite set of candidate advisers. 
Third, the structural properties of the candidate advisers inherited from the properties of the nominal adviser allow us to prove that the problem of finding $\alpha^\star$ and $\sigma_p^{\alpha^\star\star}$ can be transformed to a mean-payoff game. By that, we prove that at least one $\sigma_p^{\alpha^\star\star}$ is memoryless and hence we establish decidability of Problem~\ref{problem:main}.
\new{Finally, we discuss how the set of the candidate advisers and their associated optimal protagonist's winning strategies can be used to guide an adversary who disobeys a subset of advises provided by a least-limiting adviser.}

\subsection{Nominal adviser}
\label{sec:nominal}
The algorithm to find the nominal adviser $\alpha^0$ is summarized in Alg. \ref{alg:nominal}. It systematically finds a set of states $\losing$, from which reaching of the unsafe set $\unsafe$ cannot be avoided under any possible protagonist's and any adversary's choice of inputs. The set $\losing$ is obtained via the computation of the finite converging sequence $\unsafe = \losing^0 \subset \losing^1 \subset \ldots \subset \losing^{n-1} = \losing^n = \losing$, $n \geq 0$, where for all $0 \leq j < n$, $\losing^{j+1}$ is the set of states each of which either already belongs to $\losing^{j}$ or has all outgoing transitions leading to $\losing^{j}$ (line \ref{line:losing}). The nominal adviser $\alpha^0$ is set to forbid all transitions that lead to $\losing$ (line \ref{line:adviser}). {By construction, the algorithm terminates in at most $|S|$ iteration of the while loop (lines~\ref{line:while}--\ref{line:whileend}).}

\begin{algorithm}[!h]
\small
\caption{The nominal adviser $\alpha^0$}
\label{alg:nominal}
 \KwData{$\TS$, and unsafe set $\unsafe \subseteq S$}
 \KwResult{$\alpha^0: S_a \rightarrow 2^{U_a}$}
 \ForAll{$s_a \in S_a$}{
  $\alpha^0(s_a):= \emptyset$
 }
 \ForAll{$s_a \in \unsafe$} 
 {\label{line:unsafe}  
  $\alpha^0(s_a):= U_a^{s_a}$
 } \label{line:unsafeend}
  $\losing^0:= \unsafe$\\
  $j := 0$\\
 \While{$j = 0$ or $\losing^j \neq \losing^{j-1}$}{ \label{line:while}
   \ForAll{$s_p \in \losing^j$}{
   \ForAll{$s_a,u_a$, such that $T(s_a,u_a)=s_p$}{
   $\alpha^0(s_a):= \alpha^0(s_a) \cup \{u_a\}$ \label{line:adviser}
   } 
   }
   $\losing^{j+1}:=\losing^j \cup \{s_i \in S_i \mid \bigcup_{u_i \in U_i^{s_i}} \{T_i(s_i,u_i)\} \subseteq \losing^j, i \in \{a,p\}\}$ \label{line:losing}\\ 
   $j := j+1$
 \label{line:whileend} }
$\losing:= \losing^{j}$ 
\end{algorithm}

The following three lemmas summarize the key features of $\alpha^0$ computed according to Alg. \ref{alg:nominal}. The first two state that, if there exists a good adviser for $(\T,W_\safe)$, then the nominal adviser is good. The third states that, if the nominal adviser forbids the adversary to apply an input $u_a \in \alpha^0(s_a)$ in a state $s_a$, then there does not exist a {less limiting good adviser $\alpha'$}, such that $u_a \not \in \alpha'(s_a)$.

\vspace{-0.2cm}

\begin{lemma} If $s_\init \in \losing$ then there does not exist a good adviser for 
$(\T,W_\safe)$.
\label{lemma:1}
\end{lemma}
\begin{proof}
Suppose that $s_\init \in \losing$ and there exist a good adviser $\alpha$ for $(\T,W_\safe)$. Then there exists a non-blocking adviser restricted arena $\T^\alpha$ and a protagonist's strategy $\sigma_p^\alpha \in \Sigma_p^{\alpha}$, such that $\Plays^{\sigma_p^\alpha,\Sigma_a^{\alpha}} \subseteq W_\safe$ in $\T^{\alpha}$. Consider a  play $\pi = s_{p,1}s_{a,1}s_{p,2}s_{a,2}\ldots \in \Plays^{\sigma_p^\alpha,\Sigma_a^{\alpha}}$ in $\T^\alpha$ and note that $\pi$ does not intersect $\unsafe$.
Suppose that $s_{p,1} = s_\init \in \losing \setminus \unsafe$. Then there exists $j \geq 1$, such that $s_\init \in \losing^j$, but $s_\init \not \in \losing^{j-1}$ and $\bigcup_{u_p \in U_p^{s_\init}} \{T_p^\alpha(s_\init,u_p)\} \subseteq \losing^{j-1}$ (line~\ref{line:losing}). Thus, $s_{a,1} \in \losing^{j-1}$. Furthermore,  if $s_{a,1} \not \in \unsafe$ then $\bigcup_{u_a \in U_a^{s_{a,1}}} \{T_a^\alpha(s_{a,1},u_a)\} \subseteq \losing^{j-2}$ (line~\ref{line:losing}). Via inductive application of analogous arguments, we obtain that there exists $k \geq 1$, such that either $s_{p,k} \in \losing^0 = \unsafe$ or $s_{a,k} \in \losing^0 = \unsafe$. This contradicts the assumption that $\pi$ is winning, i.e. the assumption that $\alpha$ is a good adviser.
\end{proof}

\vspace{-0.2cm}

\begin{lemma} 
If $s_\init \not \in \losing$, then $\alpha^0$ computed by Alg.~\ref{alg:nominal} is a good adviser. 
\label{lemma:2}
\end{lemma}

{
\begin{proof}
Let $s_\init \not \in \losing$. From the construction of $\losing$ and $\alpha^0$, it follows that for all $s_p \in S_p \setminus \losing$ there exists an input $u_p$ and a state $s_a \in S_a \setminus \losing$, such that $T_a(s_p,u_p) = s_a$ (line \ref{line:losing}). Furthermore, for all $s_a \in S_a \setminus \losing$ and all $u_a \in U_a^{s_a} \setminus \alpha^0(s_a)$ it holds that $T(s_a,u_a) \in S_p \setminus \losing$ (line \ref{line:adviser}). 
Hence, there exists a play $\pi \in \Plays^{\dot{\alpha}^0}$ in $\dot \T^{\alpha^0}$, and thus there also exists a non-blocking adviser restricted arena $\T^{\alpha^0}$ and  $\sigma_p^{\alpha^0} \in \Sigma_p^{\alpha^0}$, such that any play $\pi \in \Plays^{\sigma_p^{\alpha^0},\Sigma_a^{\alpha^0}}$ in $\T^{\alpha^0}$ does not intersect $\losing$. Because $\unsafe \subseteq \losing$, it holds that $\Plays^{\sigma_p^{\alpha^0},\Sigma_a^{\alpha^0}} \subseteq W_\safe$ and adviser $\alpha^0$ is good. 
\end{proof}
}

Intuitively, Lemmas \ref{lemma:1} and \ref{lemma:2} state that the restrictions imposed by the nominal adviser $\alpha^0$ were sufficient. As a corollary, it also holds that the non-blocking nominal adviser restricted arena $\T^{\alpha^0}$ does not contain any state in $\losing$ and therefore that all plays in $\T^{\alpha^0}$ are winning. Note however, that the nominal adviser does not have to be least-limiting.
As we illustrate through the following example, imposing additional restrictions on the adversary's choices might, perhaps surprisingly, lead to the avoidance of adversary's states, where a high number of inputs are forbidden.

\vspace{-0.2cm}

\begin{example}
\new{
An example of a safety game is shown in Fig.~\ref{fig:example2}.(A). The result of the nominal adviser computation according to Alg.~\ref{alg:nominal} is illustrated in Fig.~\ref{fig:example2}.(B). Namely, $\losing = \{s_4\}$, and $\alpha^0(s_2)=\{u_{a_2}\}$, $\alpha^0(s_6) = \{u_{a_5}\}$, and $\alpha^0(s_7) = \emptyset$. There is only one protagonist's strategy $\{\sigma_p^{\alpha^0}\} = \Sigma_p^{\alpha^0}$, and it is winning $\sigma_p^{\alpha^0} \in \Omega_p^{\alpha^0}$ since $\Plays^{\sigma_p^{\alpha^0},\Sigma_a^{\alpha^0}} = \{s_1s_2s_3s_6s_3s_6\ldots, s_1s_2s_5s_7s_5s_7\ldots\}.$
The level of limitation of $\alpha^0$ is thus  
$
\lambda(\alpha^0)  =   \sup_{\sigma_a^\alpha \in \Sigma_a^\alpha} \limsup_{n \rightarrow \infty} \frac{1}{n} \sum_{j=1}^{n} \big|\alpha(\pi^{\sigma_p^\alpha,\sigma_a^\alpha}(2j))\big| = \limsup_{n \rightarrow \infty} \frac{1}{n} \Big(\big|\alpha(s_2)\big| + \sum_{j=2}^{n} \big|\alpha(s_3)\big|\Big) = 1. $
Loosely speaking, the worst-case adversary's strategy $\sigma_a$ that respects the nominal adviser $\alpha^0$ takes the play to the left-hand branch of the system.

Fig.~\ref{fig:example2}.(C) shows an alternative adviser $\alpha'$ that guides each play to the right-hand branch of the system. It is good since there is a non-blocking adviser limited arena $\T^{\alpha'}$ and the only protagonist's strategy $\sigma_p^{\alpha'} \in \Sigma_p^{\alpha'}$ on $\T^{\alpha'}$ is winning, since $\Plays^{\sigma_p^{\alpha'},\Sigma_a^{\alpha'}} = \{s_1s_2s_5s_7s_5s_7\ldots\}.$
The level of limitation of $\alpha'$ is 
$
\lambda(\alpha')  =  \limsup_{n \rightarrow \infty} \frac{1}{n} \Big(\big|\alpha(s_1)\big| + \sum_{j=2}^{n} \big|\alpha(s_5)\big|\Big) = \limsup_{n \rightarrow \infty} \frac{1}{n} \Big(\big|\alpha(s_1)\big|\Big) \ll 1.$
Hence, $\alpha'$ is less limiting than the nominal adviser $\alpha^0$.
}
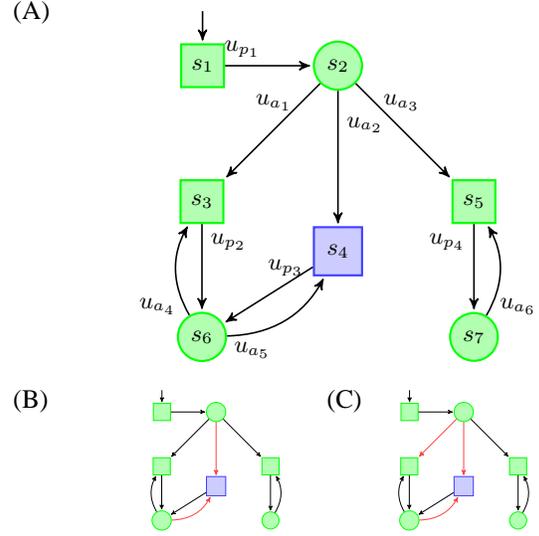
\begin{figure}[!t]
\vspace{0.2cm}
(A)
\vspace{-0.6cm}
\small
\begin{center}
\usetikzlibrary{arrows,positioning,automata,shadows,fit,shapes}
\begin{tikzpicture}[->,>=stealth',shorten >=1pt,auto, node distance=1.8cm,semithick,initial text=]

  \tikzstyle{every state}=[draw]
  \tikzstyle{a}=[circle,thick,draw=green!99,fill=green!30,minimum size=10mm]
  \tikzstyle{aunsafe}=[a,draw=blue!75,fill=blue!20]
  \tikzstyle{p}=[rectangle,thick,draw=green!99,
  			  fill=green!30,minimum size=6mm]
  \tikzstyle{punsafe}=[p,draw=blue!75,fill=blue!20]
  			  
	  \node[initial, initial where=above, p, minimum size=4ex] (q1) {$s_1$}; 
	  \node[a, minimum size=4.5ex] (q2) [right of = q1] {$s_2$}; 
      \node[p, minimum size=4ex] (q21) [below left = of q2] {$s_3$}; 
	  \node[punsafe, minimum size=4.5ex] (q3) [below = of q2] {$s_4$}; 
	  \node[p, minimum size=4ex] (q4) [below right = of q2] {$s_5$};
      \node[a, minimum size=4.5ex] (q41) [below of = q21] {$s_6$};  
	  \node[a, minimum size=4ex] (q5) [below of = q4] {$s_7$}; 
	
\begin{scope}[every node/.style={scale=1}]
  \path (q1) edge node [pos=0.2] {$u_{p_1}$} (q2);
  \path (q2) edge  node [pos=0.2,left] {$u_{a_1}$} (q21);
  \path (q2) edge  node [pos=0.25,right] {$u_{a_2}$} (q3);
  \path (q2) edge  node [pos=0.2,right] {$u_{a_3}$} (q4);
  \path (q21) edge node [pos=0.2] {$u_{p_2}$} (q41);
  \path (q4) edge  node  [pos=0.2, left] {$u_{p_4}$} (q5);
  \path (q3) edge node  [pos=0.3, above] {$u_{p_3}$} (q41);
  \path (q41) [bend left] edge node [pos=0.1, left] {$u_{a_4}$} (q21);
  \path (q41) [bend right] edge node [pos=0.2,below] {$u_{a_5}$} (q3);
  \path (q5) [bend right] edge node [pos=0.1,right] {$u_{a_6}$} (q4);
 \end{scope}
\end{tikzpicture} \\
\end{center}
\normalsize{(B) \hspace{0.4\linewidth} (C) }
\vspace{-0.6cm}
\small
\begin{center}
\scalebox{0.4}{
\usetikzlibrary{arrows,positioning,automata,shadows,fit,shapes}
\begin{tikzpicture}[->,>=stealth',shorten >=1pt,auto, node distance=1.8cm,semithick,initial text=]

  \tikzstyle{every state}=[draw]
  \tikzstyle{a}=[circle,thick,draw=green!99,fill=green!30,minimum size=10mm]
  \tikzstyle{aunsafe}=[a,draw=blue!75,fill=blue!20]
  \tikzstyle{p}=[rectangle,thick,draw=green!99,
  			  fill=green!30,minimum size=6mm]
  \tikzstyle{punsafe}=[p,draw=blue!75,fill=blue!20]
  			  
	  \node[initial, initial where=above, p, minimum size=4ex] (q1) {}; 
	  \node[a, minimum size=4.5ex] (q2) [right of = q1] {}; 
      \node[p, minimum size=4ex] (q21) [below left = of q2] {}; 
	  \node[punsafe, minimum size=4.5ex] (q3) [below = of q2] {}; 
	  \node[p, minimum size=4ex] (q4) [below right = of q2] {};
      \node[a, minimum size=4.5ex] (q41) [below of = q21] {};  
	  \node[a, minimum size=4ex] (q5) [below of = q4] {}; 
	
\begin{scope}[every node/.style={scale=1}]
  \path (q1) edge node [pos=0.2] {} (q2);
  \path (q2) edge  node [pos=0.2,left] {} (q21);
  \path (q2) edge  [thick,red!75]   node [pos=0.25,right] {} (q3);
  \path (q2) edge  node [pos=0.2,right] {} (q4);
  \path (q21) edge node [pos=0.2] {} (q41);
  \path (q4) edge  node  [pos=0.2, left] {} (q5);
  \path (q3) edge node  [pos=0.3, above] {} (q41);
  \path (q41) [bend left] edge node [pos=0.1, left] {} (q21);
  \path (q41) [bend right,thick,red!75] edge node [pos=0.2,below] {} (q3);
  \path (q5) [bend right] edge node [pos=0.1,right] {} (q4);
 \end{scope}
\end{tikzpicture}
\hspace{3.5cm}
\begin{tikzpicture}[->,>=stealth',shorten >=1pt,auto, node distance=1.8cm,semithick,initial text=]

  \tikzstyle{every state}=[draw]
  \tikzstyle{a}=[circle,thick,draw=green!99,fill=green!30,minimum size=10mm]
  \tikzstyle{aunsafe}=[a,draw=blue!75,fill=blue!20]
  \tikzstyle{p}=[rectangle,thick,draw=green!99,
  			  fill=green!30,minimum size=6mm]
  \tikzstyle{punsafe}=[p,draw=blue!75,fill=blue!20]
  			  
	  \node[initial, initial where=above, p, minimum size=4ex] (q1) {}; 
	  \node[a, minimum size=4.5ex] (q2) [right of = q1] {}; 
      \node[p, minimum size=4ex] (q21) [below left = of q2] {}; 
	  \node[punsafe, minimum size=4.5ex] (q3) [below = of q2] {}; 
	  \node[p, minimum size=4ex] (q4) [below right = of q2] {};
      \node[a, minimum size=4.5ex] (q41) [below of = q21] {};  
	  \node[a, minimum size=4ex] (q5) [below of = q4] {}; 
	
\begin{scope}[every node/.style={scale=1}]
  \path (q1) edge node [pos=0.2] {} (q2);
  \path (q2) edge [thick,red!75]   node [pos=0.2,left] {} (q21);
  \path (q2) edge [thick,red!75]  node [pos=0.25,right] {} (q3);
  \path (q2) edge  node [pos=0.2,right] {} (q4);
  \path (q21) edge node [pos=0.2] {} (q41);
  \path (q4) edge  node  [pos=0.2, left] {} (q5);
  \path (q3) edge node  [pos=0.3, above] {} (q41);
  \path (q41) [bend left] edge node [pos=0.1, left] {} (q21);
  \path (q41) [bend right,thick,red!75] edge node [pos=0.2,below] {} (q3);
  \path (q5) [bend right] edge node [pos=0.1,right] {} (q4);
 \end{scope}
\end{tikzpicture}

}
\end{center}
\vspace{-0.2cm}
\caption{(A) An example of a game arena with a safety winning condition. The protagonist's and adversary's states are illustrated as squares and circles, respectively. The safe set $\safe$ is in green, the unsafe set $\unsafe$ in blue. Transitions are depicted as arrows between them and they are labeled with the respective inputs that trigger them. (B) shows the nominal adviser $\alpha^0$ and $\losing$ via marking the forbidden transitions and the states in $\losing$ in red. (C) shows an alternative adviser $\alpha'$ that is also good and less limiting than $\alpha^0$.}
\label{fig:example2}
\vspace{-0.7cm}
\end{figure}
\end{example}

\vspace{-0.2cm}

\begin{lemma}
{
Consider an adviser $\alpha'$ for $(\T,W_\safe)$ and suppose that there exists a state $s_a \in S_a$ and $u_a \in U_a$, such that $u_a \in \alpha^0(s_a)$ and $u_a \not \in \alpha'(s_a)$. Then $\alpha'$ is either not good or at least as limiting as the nominal adviser $\alpha^0$, i.e. $\lambda(\alpha^0) \leq \lambda(\alpha')$.
}
\label{lemma:nominal}
\end{lemma}

{
\begin{proof}
The proof is lead by contradiction. Consider an adviser $\alpha'$ for $(\T,W_\safe)$. Suppose that there exists a state $s_a \in S_a$ and $u_a \in U_a$, such that $u_a \in \alpha^0(s_a)$ and $u_a \not \in \alpha'(s_a)$ and $\alpha'$ is good. Furthermore, let $\Omega_p^{\alpha'}$ be the set of protagonist's winning strategies on the non-blocking adviser restricted arena $\T^{\alpha'}$. and assume that $\alpha'$ is less limiting that $\alpha^0$, i.e. that $\lambda(\alpha')< \lambda(\alpha^0) $. Then from the definition of $\lambda$ in~Eq. \eqref{eq:alpha}, there exists a protagonist's strategy $\sigma_p^{\alpha'} \in \Omega_p^{\alpha'}$, such that $\gamma(\sigma_p^{\alpha'}) < \lambda(\alpha^0)$.
Henceforth, there also exists a winning play $\pi = s_{p,1}s_{a,1}s_{p,2}s_{a,2} \ldots \in \Plays^{\sigma_p^{\alpha'},\Sigma_a^{\alpha'}}$ on $\T^{\alpha'}$ with the property that for some $k\geq 1$, $s_{p,k+1} \in T(s_{a,k},u_a)$, where $u_a \not \in \alpha'(s_{a,k})$ and $u_a \in \alpha^0(s_{a,k})$. If such a winning play does not exist, it holds that $\gamma(\sigma_p^{\alpha'}) \geq \lambda(\alpha^0)$, which contradicts the assumption that $\alpha'$ is less limiting than $\alpha^0$. Since $u_a \in \alpha^0(s_{a,k})$, it holds $s_{p,k+1} \in \losing$ by construction (line~\ref{line:adviser}). Either $s_{p,k+1} \in \unsafe$, which directly contradicts the assumption that $\alpha'$ is good, or $s_{p,k+1} \in \losing^j$, for some $j \geq 1$. From the iterative construction of $\losing$, we obtain $s_{a,k+1} \in \losing^{j-1}$ and if $s_{a,k+1} \not \in \unsafe$, then $s_{p,k+2} \in \losing^{j-2}$. By inductive reasoning it follows that there exists $\ell \geq k+1$, such that either $s_{p,\ell} \in \losing^0 = \unsafe$, or $s_{a,\ell} \in \losing^0 = \unsafe$. This contradicts the assumption that $\pi$ is winning, i.e. the assumption that $\alpha'$ is good.
\end{proof}
}

Thanks to Lemma \ref{lemma:nominal}, we know 
that there exists a good adviser $\alpha^\star$ that is least-limiting and builds on the nominal one in the following sense: $\alpha^0(s_a) \subseteq \alpha^\star (s_a)$, for all $s_a \in S_a$. 
\new{Whereas following the nominal adviser is essential for maintaining the system safety, following the additional restrictions suggested by $\alpha^\star$ can be perceived as a weak form of advice. If this advice is not respected by the adversary, safety is not necessarily going to be violated, however, in order to maintain safety, the adversary might need to obey further, more limiting advises. We will discuss later on in Sec.~\ref{sec:application} how to use both the combination of a least-limiting adviser and the nominal one in order to guide the adversary during the system execution (the play on the game arena).}

\subsection{Least-limiting solution}
\label{sec:least}

Let $\dot A_{\mathit{cand}}$ denote the finite set of candidate advisers obtained from the nominal adviser $\alpha^0$, 
$\dot A_{\mathit{cand}} = \{\alpha \mid \alpha^0(s_a) \subseteq \alpha (s_a) , \text{ for all }  s_a \in S_a \}.$
\new{Note that $\alpha \in \dot{A}_\mathit{cand}$ does not have to be good since it might not allow for an existence of a non-blocking adviser restricted arena $\T^\alpha$. As outlined in Sec.~\ref{sec:games}, it can be however decided whether $\dot \T^\alpha$ from Def.~\ref{def:restricted} has an equivalent non-blocking arena $\T^\alpha$. Building on ideas from Lemmas~\ref{lemma:1} and~\ref{lemma:2}, we can easily see that the existence of non-blocking adviser restricted arena $\T^\alpha$ also implies the existence of a protagonist's winning strategy $\sigma_p^\alpha \in \Omega_p^\alpha$. In fact, because states from $\losing$ were removed from $ \T^{\alpha^0}$ (lines~\ref{line:unsafe}--\ref{line:unsafeend}, \ref{line:while}--\ref{line:whileend} of Alg.~\ref{alg:nominal}), all plays in $\T^\alpha$ are winning and $\Sigma_p^\alpha = \Omega_p^\alpha$.
} \begin{align}\cand = \{\alpha \in \dot A_{\mathit{cand}} \mid \alpha \text{ is a good adviser} \}.
\label{eq:cand}
\end{align}

From Lemma~\ref{lemma:nominal} and the construction of $\cand$, at least one least-limiting good adviser belongs to $\cand$. In the remainder of the solution, we focus on solving the following sub-problem for each $\alpha \in \cand$.



\vspace{-0.2cm}

\new{
\begin{problem}
Consider a good adviser $\alpha \in \cand$. Find $\lambda(\alpha)$ and an optimal protagonist's winning strategy $\sigma_p^{\alpha\star}$ with $\gamma(\sigma_p^{\alpha\star})= \inf_{\sigma_p^\alpha \in \Omega_p^\alpha}\gamma(\sigma_p^\alpha) = \inf_{\sigma_p^\alpha \in \Sigma_p^\alpha}\gamma(\sigma_p^\alpha) .$ 
\label{problem:2}
\end{problem}

We propose to translate Problem \ref{problem:2} to finding an optimal strategy to a mean-payoff game on a modified arena $\widetilde \T^\alpha$:

\vspace{-0.2cm}

\begin{definition}[Mean-payoff game arena $\widetilde \T^\alpha$] Given a non-blocking adviser restricted arena $\T^\alpha = (S^
\alpha, \langle S_p^\alpha,S_a^\alpha \rangle, s_\init, U_p, U_a, T_p^\alpha \cup T_a^\alpha)$, we define the mean-payoff game arena $\widetilde{\T}^\alpha = (\mathcal T^\alpha, w),$ where 
for all $\widetilde T_p(s_p,u_p) = s_a$, $w(s_p,s_a) = - |\alpha(s_a)|$ and for all $\widetilde T_a(s_a,u_a) = s_p$, $w(s_a,s_p) = 0$.
\label{def:mean}
\end{definition}

\vspace{-0.2cm}

\begin{lemma}
Problem~\ref{problem:2} reduces to the problem of optimal strategy synthesis for the mean-payoff game $\widetilde{\T}^\alpha$.
\end{lemma}

\begin{proof}
The optimal strategy  $\widetilde{\sigma}_p^{\alpha\star}$ for the mean-payoff game $\widetilde{\T^\alpha}$ obtained e.g., by the algorithm from~\cite{brim} has the  value $ \nu(\widetilde{\sigma}_p^{\alpha\star}) = $\vspace{-0.5cm}
\begin{align*}
& \sup_{\sigma_p^\alpha \in \Sigma_p^\alpha}\inf_{\sigma_a^\alpha \in \Sigma_a^\alpha} \liminf_{n \rightarrow \infty} \frac{1}{n} \sum_{j=1}^{n} w(\pi^{\sigma_p^\alpha,\sigma_a^\alpha}(j),\pi^{\sigma_p^\alpha,\sigma_a^\alpha}(j+1) )=\\
&\inf_{\sigma_p^\alpha \in \Sigma_p^\alpha}\sup_{\sigma_a^\alpha \in \Sigma_a^\alpha} \limsup_{n \rightarrow \infty} \frac{1}{n} \sum_{j=1}^{n} -w(\pi^{\sigma_p^\alpha,\sigma_a^\alpha}(j),\pi^{\sigma_p^\alpha,\sigma_a^\alpha}(j+1) )=\\
& \inf_{\sigma_p^\alpha \in \Sigma_p^\alpha}\sup_{\sigma_a^\alpha \in \Sigma_a} \limsup_{n \rightarrow \infty} \frac{1}{n} \sum_{j=1}^{n} |\alpha(\pi^{\sigma_p^\alpha,\sigma_a^\alpha}(2j)| = \lambda(\alpha).
\end{align*}

Furthermore, as noted above $\Sigma_p^\alpha = \Omega^\alpha_p$ and hence the proof is complete.
\end{proof}

It has been shown in~\cite{positional} that in mean-payoff games, memoryless strategies suffice to achieve the optimal value. In fact, using the algorithm from~\cite{brim}, the strategy $\widetilde \sigma_p^{\alpha\star}$ takes the form of a memoryless strategy $\widetilde \varsigma_p^{\alpha\star}: S_p^\alpha \rightarrow U_p$. 

\vspace{-0.1cm}

\subsection{Overall solution}

\vspace{-0.1cm}

We summarize how the algorithms from Sec.~\ref{sec:nominal} and Sec.~\ref{sec:least} serve in finding a solution to Problem~\ref{problem:main}.
1) The nominal adviser $\alpha^0$ is built according to Alg.~\ref{alg:nominal}. If there does not exist a non-blocking adviser restricted arena $\T^{\alpha^0}$, then there does not exist a solution to Problem~\ref{problem:main}.
2) The set of candidate advisers $\cand$ is built according to Eq.~\eqref{eq:cand}.
3) For each candidate adviser $\alpha \in \cand$, the value $\lambda(\alpha)$  and the memoryless optimal protagonist's winning strategy $\varsigma_p^{\alpha \star} \in \Omega_p^\alpha$ are computed through the translation to a mean-payoff game optimal strategy synthesis according to Def.~\ref{def:mean}.
4) An adviser $\alpha^\star \in \cand$ with $\lambda(\alpha^\star) = \inf_{\alpha\in \cand} \lambda(\alpha)$ together with its associated optimal strategy $\varsigma_p^{\alpha^\star\star}$ are the solution to Problem~\ref{problem:main}.

}

\subsection{Guided system execution}
\label{sec:application}

\new{Finally, we discuss how the set of good advisers $\cand$ can be used to guide the adversary on-the-fly during the system execution. 
Given an adviser $\alpha \in \cand$, let us call the fact that $u_a \in \alpha(s_a)$ an \emph{advise}. We distinguish two types of advises, \emph{hard} and \emph{soft}. 
Hard advises are the ones imposed by the nominal adviser, 
$u_a \in \alpha^0(s_a)$, while soft are the remaining ones that can be violated without jeopardizing the system safety. The goal of the guided execution is to permit the adversary to disobey a soft advise and react to this event via a switch to another, possibly more limiting adviser that does not contain this soft advise.
Let $\preceq$ be a partial ordering on the set $\cand$, where $\alpha \preceq \alpha'$ if $\alpha(s_a) \subseteq \alpha'(s_a)$, for all $s_a \in S_a$. Hence, for the nominal adviser $\alpha^0$, it holds that $\alpha^0 \preceq \alpha$, for all $\alpha \in \cand$. 

The system execution that corresponds to a play in $\T$ proceeds as follows:
1) The system starts at the initial state $s_\curr=s_\init$ with the current adviser being least-limiting adviser $\alpha_\curr = \alpha^\star$ and the current protagonist's strategy being the memoryless winning strategy $\varsigma_{p,\curr} = \varsigma_p^{\alpha^\star\star}$. 
2) The input $\varsigma_{p,\curr}(s_\curr)$ is applied by the protagonist and the system changes its current state $s_\curr$ according to $T_p$. The current state belongs to the adversary.
3) $\alpha_\curr(s_\curr)$ is provided. The adversary chooses an input $u_a \in U_a^{s_\curr}$. 
a) If $u_a \not \in \alpha_\curr(s_\curr)$, then the system updates its state $s_\curr$ according to $T_a$ and proceeds with step $2$. 
b) If $u_a \in \alpha^0(s_\curr)$ then hard advise is disobeyed and system safety will be unavoidably violated and the system needs to stop immediately.
c) If $u_a \in \alpha_\curr(s_\curr)$, but $u_a \not \in \alpha^0(s_\curr)$, then only a soft advise is disobeyed. The current adviser $\alpha_\curr$ is updated to $\alpha'$, with the property that $\lambda(\alpha') = \inf_{\alpha \in A_\preceq} \lambda(\alpha)$, where $A_\preceq = \{\alpha \in \cand \mid \alpha \preceq \alpha_\curr\}$ and the current protagonist's strategy $\varsigma_{p,\curr}$ is updated to $\varsigma_{p}^{\alpha'\star}$. The current state $s_\curr$ is updated according to $T_a$ and the system proceeds with step 2).

}

\vspace{-0.2cm}

\begin{example}
\new{

Consider the safety game in Fig.~\ref{fig:example3}.(A). The result of the nominal adviser computation according to Alg.~\ref{alg:nominal} is illustrated in Fig.~\ref{fig:example3}.(B). Namely, $\alpha^0(s_2) = \{u_{a_2}\}$,
$\alpha^0(s_6) = \emptyset$,
$\alpha^0(s_7) = \{u_{a_5}\}$,  $\alpha^0(s_8) = \{u_{a_7},u_{a_8}\}$, and $\alpha^0(s_{11}) = u_{a_9}$. The states in $\losing$ are marked in red. Fig. \ref{fig:example3}.(C) shows the non-blocking adviser restricted arena $\T^{\alpha^0}$ with the removed states and transitions in light grey. 
The corresponding optimal protagonist's winning strategy $\varsigma_p^{\alpha^0\star}$ in $\T^{\alpha^0}$ is highlighted in green in Fig.~\ref{fig:example3}.(B), i.e. $\varsigma_p^{\alpha^0\star}(s_1) = u_{p_1}$, $\varsigma_p^{\alpha^0\star}(s_3) = u_{p_2}$, $\varsigma_p^{\alpha^0\star}(s_5) = u_{p_5}$, and $\varsigma_p^{\alpha^0\star}(s_9) = u_{p_6}$.
The level of limitation of $\alpha^0$ is $\lambda(\alpha^0) = \limsup_{n\rightarrow \infty} \frac{1}{n} (|\alpha^0(s_2)| + \sum_{j=2}^n |\alpha^0(s_8)|)$ = $\limsup_{n\rightarrow \infty} \frac{1}{n} (2n-1)$. 
Fig. \ref{fig:example3}.(D) shows least-limiting adviser $\alpha^\star$. As opposed to $\alpha^0$, $\alpha^\star(s_2) = \{u_{a_2},u_{a_3}\}$, where the advise $u_{a_3} \in \alpha^\star(s_2)$ (in magenta) is soft. Fig.~\ref{fig:example3}.(E) illustrates the non-blocking adviser restricted arena $\T^{\alpha^\star}$. The optimal protagonist's winning strategy is the only protagonist's strategy in $\T^{\alpha^\star}$. The level of limitation of $\alpha^\star$ is $\lambda(\alpha^\star) = \limsup_{n\rightarrow \infty} \frac{1}{n} ( |\alpha^\star(s_2)|  + \sum_{j=2}^n |\alpha^0(s_6)|)$ = $\limsup_{n\rightarrow \infty} \frac{1}{n}  < \lambda(\alpha^0)$.
There exist more good advisers $\alpha' \in \cand$. For each of them, either $\lambda(\alpha') = \lambda(\alpha^0)$ or $\lambda(\alpha') = \lambda(\alpha^\star)$.

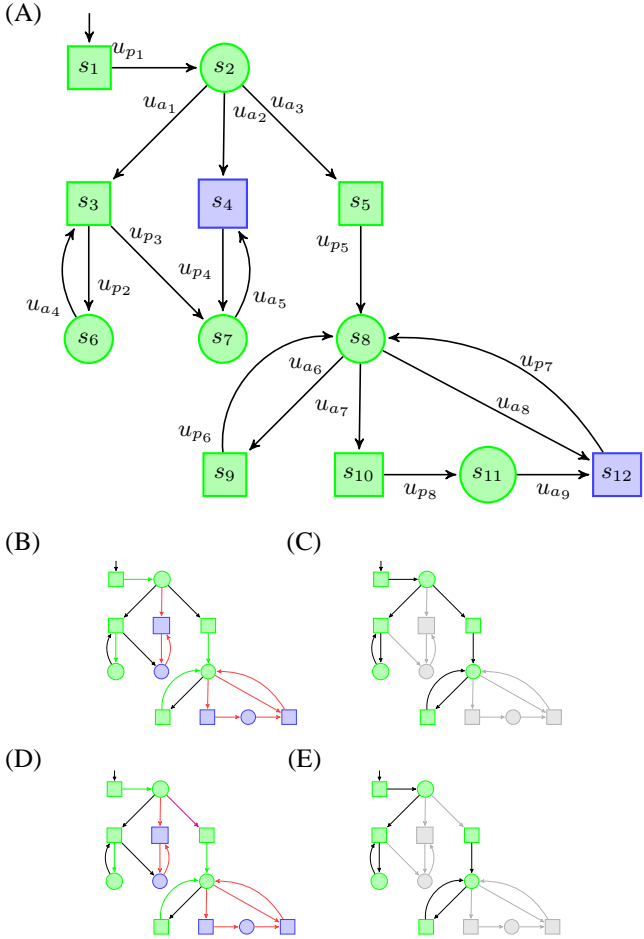
\begin{figure}[!t]
\vspace{0.2cm}
(A)
\vspace{-0.6cm}
\small
\begin{center}
\usetikzlibrary{arrows,positioning,automata,shadows,fit,shapes}
\begin{tikzpicture}[->,>=stealth',shorten >=1pt,auto, node distance=1.8cm,semithick,initial text=]

  \tikzstyle{every state}=[draw]
  \tikzstyle{a}=[circle,thick,draw=green!99,fill=green!30,minimum size=10mm]
  \tikzstyle{aunsafe}=[a,draw=blue!75,fill=blue!20]
  \tikzstyle{p}=[rectangle,thick,draw=green!99,
  			  fill=green!30,minimum size=6mm]
  \tikzstyle{punsafe}=[p,draw=blue!75,fill=blue!20]
  			  
	  \node[initial, initial where=above, p, minimum size=4ex] (q1) {$s_1$}; 
	  \node[a, minimum size=4.5ex] (q2) [right of = q1] {$s_2$}; 
      \node[p, minimum size=4ex] (q21) [below left = of q2] {$s_3$}; 
	  \node[punsafe, minimum size=4.5ex] (q3) [left = 1.2cm of q4] {$s_4$}; 
	  \node[p, minimum size=4ex] (q4) [below right = of q2] {$s_5$};
      \node[a, minimum size=4.5ex] (q41) [below of = q21] {$s_6$};  
 	   \node[a, minimum size=4ex] (q6) [below of = q3] {$s_7$};
	   \node[a, minimum size=4ex] (q7) [below of = q4] {$s_8$};  
	   \node[p, minimum size=4ex] (q8) [below left = of q7] {$s_{9}$}; 
	   \node[p, minimum size=4ex] (q9) [right = 1.15cm of q8] {$s_{10}$};
	   \node[a, minimum size=4ex] (q10) [right = 1cm of q9] {$s_{11}$};
	   \node[punsafe, minimum size=4ex] (q11) [right = 1cm of q10] {$s_{12}$};
\begin{scope}[every node/.style={scale=1}]
  \path (q1) edge node [pos=0.2] {$u_{p_1}$} (q2);
  \path (q2) edge  node [pos=0.2,left] {$u_{a_1}$} (q21);
  \path (q2) edge  node [pos=0.25,right] {$u_{a_2}$} (q3);
  \path (q2) edge  node [pos=0.2,right] {$u_{a_3}$} (q4);
  \path (q21) edge node [pos=0.7] {$u_{p_2}$} (q41);
  \path (q4) edge  node  [pos=0.2, left] {$u_{p_5}$} (q5);
  \path (q21) edge node  [pos=0.1, right] {$u_{p_3}$} (q6);
  \path (q41) [bend left] edge node [pos=0.1, left] {$u_{a_4}$} (q21);
  \path (q6) [bend right] edge node [pos=0.2,right] {$u_{a_5}$} (q3);
  \path (q3) edge node [pos=0.5,left] {$u_{p_4}$} (q6);
  \path (q7) edge node [pos=0.1,left] {$u_{a_6}$} (q8);
  \path (q7) edge node [pos=0.5,left] {$u_{a_7}$} (q9);
  \path (q7) edge node [pos=0.5,right] {$u_{a_8}$} (q11);
  \path (q11) edge [bend right] node [pos=0.5,right] {$u_{p_7}$} (q7);
  \path (q9) edge node [pos=0.5,below] {$u_{p_8}$} (q10);
  \path (q10) edge node [pos=0.5,below] {$u_{a_9}$} (q11);
   \path (q8) edge [bend left = 50] node [pos=0.1,left] {$u_{p_6}$} (q7);

 \end{scope}
\end{tikzpicture} \\
\end{center}
\normalsize{(B) \hspace{0.35\linewidth} (C) }
\vspace{-0.2cm}
\small
\begin{center}
\scalebox{0.34}{
\usetikzlibrary{arrows,positioning,automata,shadows,fit,shapes}
\begin{tikzpicture}[->,>=stealth',shorten >=1pt,auto, node distance=1.8cm,semithick,initial text=]

  \tikzstyle{every state}=[draw]
  \tikzstyle{a}=[circle,thick,draw=green!99,fill=green!30,minimum size=10mm]
  \tikzstyle{aunsafe}=[a,draw=blue!75,fill=blue!20]
  \tikzstyle{p}=[rectangle,thick,draw=green!99,
  			  fill=green!30,minimum size=6mm]
  \tikzstyle{punsafe}=[p,draw=blue!75,fill=blue!20]
  \tikzstyle{fp}=[p,draw=blue!75,fill=blue!20]
  \tikzstyle{fa}=[a,draw=blue!75,fill=blue!20]  			  
	  \node[initial, initial where=above, p, minimum size=4ex] (q1) {}; 
	  \node[a, minimum size=4.5ex] (q2) [right of = q1] {}; 
      \node[p, minimum size=4ex] (q21) [below left = of q2] {}; 
	  \node[fp, minimum size=4.5ex] (q3) [left = 1.2cm of q4] {}; 
	  \node[p, minimum size=4ex] (q4) [below right = of q2] {};
      \node[a, minimum size=4.5ex] (q41) [below of = q21] {};  
 	   \node[fa, minimum size=4ex] (q6) [below of = q3] {};
	   \node[a, minimum size=4ex] (q7) [below of = q4] {};  
	   \node[p, minimum size=4ex] (q8) [below left = of q7] {}; 
	   \node[fp, minimum size=4ex] (q9) [right = 1.15cm of q8] {};
	   \node[fa, minimum size=4ex] (q10) [right = 1cm of q9] {};
	   \node[fp, minimum size=4ex] (q11) [right = 1cm of q10] {};
\begin{scope}[every node/.style={scale=1}]
  \path (q1) edge [green!99,thick] node [pos=0.2] {} (q2);
  \path (q2) edge  node [pos=0.2,left] {} (q21);
  \path (q2) edge [red!75,thick] node [pos=0.25,right] {} (q3);
  \path (q2) edge  node [pos=0.2,right] {} (q4);
  \path (q21) edge [green!99,thick] node [pos=0.7] {} (q41);
  \path (q4) [green!99,thick] edge  node  [pos=0.2, left] {} (q5);
  \path (q21) edge node  [pos=0.1, right] {} (q6);
  \path (q41) [bend left] edge node [pos=0.1, left] {} (q21);
  \path (q6) [red!75,thick,bend right] edge node [pos=0.2,right] {} (q3);
  \path (q3) edge [red!75,thick]  node [pos=0.5,left] {} (q6);
  \path (q7) edge node [pos=0.1,left] {} (q8);
  \path (q7) edge [red!75,thick] node [pos=0.5,left] {} (q9);
  \path (q7) edge [red!75,thick] node [pos=0.5,right] {} (q11);
  \path (q11) edge [red!75,thick,bend right] node [pos=0.5,right] {} (q7);
  \path (q9) edge [red!75,thick]  node [pos=0.5,below] {} (q10);
  \path (q10) edge [red!75,thick]  node [pos=0.5,below] {} (q11);
   \path (q8) edge [green!99,thick,bend left = 50] node [pos=0.1,left] {} (q7);

 \end{scope}
\end{tikzpicture} 

\hspace{2.5cm}

\usetikzlibrary{arrows,positioning,automata,shadows,fit,shapes}
\begin{tikzpicture}[->,>=stealth',shorten >=1pt,auto, node distance=1.8cm,semithick,initial text=]

  \tikzstyle{every state}=[draw]
  \tikzstyle{a}=[circle,thick,draw=green!99,fill=green!30,minimum size=10mm]
  \tikzstyle{aunsafe}=[a,draw=blue!75,fill=blue!20]
  \tikzstyle{p}=[rectangle,thick,draw=green!99,
  			  fill=green!30,minimum size=6mm]
  \tikzstyle{punsafe}=[p,draw=blue!75,fill=blue!20]
  \tikzstyle{fp}=[p,draw=black!30,fill=black!10]
  \tikzstyle{fa}=[a,draw=black!30,fill=black!10]  			  
	  \node[initial, initial where=above, p, minimum size=4ex] (q1) {}; 
	  \node[a, minimum size=4.5ex] (q2) [right of = q1] {}; 
      \node[p, minimum size=4ex] (q21) [below left = of q2] {}; 
	  \node[fp, minimum size=4.5ex] (q3) [left = 1.2cm of q4] {}; 
	  \node[p, minimum size=4ex] (q4) [below right = of q2] {};
      \node[a, minimum size=4.5ex] (q41) [below of = q21] {};  
 	   \node[fa, minimum size=4ex] (q6) [below of = q3] {};
	   \node[a, minimum size=4ex] (q7) [below of = q4] {};  
	   \node[p, minimum size=4ex] (q8) [below left = of q7] {}; 
	   \node[fp, minimum size=4ex] (q9) [right = 1.15cm of q8] {};
	   \node[fa, minimum size=4ex] (q10) [right = 1cm of q9] {};
	   \node[fp, minimum size=4ex] (q11) [right = 1cm of q10] {};
\begin{scope}[every node/.style={scale=1}]
  \path (q1) edge  node [pos=0.2] {} (q2);
  \path (q2) edge  node [pos=0.2,left] {} (q21);
  \path (q2) edge [black!30,thick] node [pos=0.25,right] {} (q3);
  \path (q2) edge  node [pos=0.2,right] {} (q4);
  \path (q21) edge node [pos=0.7] {} (q41);
  \path (q4) edge  node  [pos=0.2, left] {} (q5);
  \path (q21) edge  [black!30] node  [pos=0.1, right] {} (q6);
  \path (q41) [bend left] edge node [pos=0.1, left] {} (q21);
  \path (q6) [black!30,bend right] edge node [pos=0.2,right] {} (q3);
  \path (q3) edge [black!30]  node [pos=0.5,left] {} (q6);
  \path (q7) edge node [pos=0.1,left] {} (q8);
  \path (q7) edge [black!30] node [pos=0.5,left] {} (q9);
  \path (q7) edge [black!30] node [pos=0.5,right] {} (q11);
  \path (q11) edge [black!30,bend right] node [pos=0.5,right] {} (q7);
  \path (q9) edge [black!30]  node [pos=0.5,below] {} (q10);
  \path (q10) edge [black!30]  node [pos=0.5,below] {} (q11);
   \path (q8) edge [bend left = 50] node [pos=0.1,left] {} (q7);

 \end{scope}
\end{tikzpicture}}
\end{center}
\normalsize{(D) \hspace{0.35\linewidth} (E) }
\vspace{-0.3cm}
\small
\begin{center}
\scalebox{0.34}{
\usetikzlibrary{arrows,positioning,automata,shadows,fit,shapes}
\begin{tikzpicture}[->,>=stealth',shorten >=1pt,auto, node distance=1.8cm,semithick,initial text=]

  \tikzstyle{every state}=[draw]
  \tikzstyle{a}=[circle,thick,draw=green!99,fill=green!30,minimum size=10mm]
  \tikzstyle{aunsafe}=[a,draw=blue!75,fill=blue!20]
  \tikzstyle{p}=[rectangle,thick,draw=green!99,
  			  fill=green!30,minimum size=6mm]
  \tikzstyle{punsafe}=[p,draw=blue!75,fill=blue!20]
  \tikzstyle{fp}=[p,draw=blue!75,fill=blue!20]
  \tikzstyle{fa}=[a,draw=blue!75,fill=blue!20]  			  
	  \node[initial, initial where=above, p, minimum size=4ex] (q1) {}; 
	  \node[a, minimum size=4.5ex] (q2) [right of = q1] {}; 
      \node[p, minimum size=4ex] (q21) [below left = of q2] {}; 
	  \node[fp, minimum size=4.5ex] (q3) [left = 1.2cm of q4] {}; 
	  \node[p, minimum size=4ex] (q4) [below right = of q2] {};
      \node[a, minimum size=4.5ex] (q41) [below of = q21] {};  
 	   \node[fa, minimum size=4ex] (q6) [below of = q3] {};
	   \node[a, minimum size=4ex] (q7) [below of = q4] {};  
	   \node[p, minimum size=4ex] (q8) [below left = of q7] {}; 
	   \node[fp, minimum size=4ex] (q9) [right = 1.15cm of q8] {};
	   \node[fa, minimum size=4ex] (q10) [right = 1cm of q9] {};
	   \node[fp, minimum size=4ex] (q11) [right = 1cm of q10] {};
\begin{scope}[every node/.style={scale=1}]
  \path (q1) edge [green!99,thick] node [pos=0.2] {} (q2);
  \path (q2) edge  node [pos=0.2,left] {} (q21);
  \path (q2) edge [red!75,thick] node [pos=0.25,right] {} (q3);
  \path (q2) edge [magenta,thick] node [pos=0.2,right] {} (q4);
  \path (q21) edge [green!99,thick] node [pos=0.7] {} (q41);
  \path (q4) [green!99,thick] edge  node  [pos=0.2, left] {} (q5);
  \path (q21) edge node  [pos=0.1, right] {} (q6);
  \path (q41) [bend left] edge node [pos=0.1, left] {} (q21);
  \path (q6) [red!75,thick,bend right] edge node [pos=0.2,right] {} (q3);
  \path (q3) edge [red!75,thick]  node [pos=0.5,left] {} (q6);
  \path (q7) edge node [pos=0.1,left] {} (q8);
  \path (q7) edge [red!75,thick] node [pos=0.5,left] {} (q9);
  \path (q7) edge [red!75,thick] node [pos=0.5,right] {} (q11);
  \path (q11) edge [red!75,thick,bend right] node [pos=0.5,right] {} (q7);
  \path (q9) edge [red!75,thick]  node [pos=0.5,below] {} (q10);
  \path (q10) edge [red!75,thick]  node [pos=0.5,below] {} (q11);
   \path (q8) edge [green!99,thick,bend left = 50] node [pos=0.1,left] {} (q7);

 \end{scope}
\end{tikzpicture} 

\hspace{2.5cm}
\usetikzlibrary{arrows,positioning,automata,shadows,fit,shapes}
\begin{tikzpicture}[->,>=stealth',shorten >=1pt,auto, node distance=1.8cm,semithick,initial text=]

  \tikzstyle{every state}=[draw]
  \tikzstyle{a}=[circle,thick,draw=green!99,fill=green!30,minimum size=10mm]
  \tikzstyle{aunsafe}=[a,draw=blue!75,fill=blue!20]
  \tikzstyle{p}=[rectangle,thick,draw=green!99,
  			  fill=green!30,minimum size=6mm]
  \tikzstyle{punsafe}=[p,draw=blue!75,fill=blue!20]
  \tikzstyle{fp}=[p,draw=black!30,fill=black!10]
  \tikzstyle{fa}=[a,draw=black!30,fill=black!10]  			  
	  \node[initial, initial where=above, p, minimum size=4ex] (q1) {}; 
	  \node[a, minimum size=4.5ex] (q2) [right of = q1] {}; 
      \node[p, minimum size=4ex] (q21) [below left = of q2] {}; 
	  \node[fp, minimum size=4.5ex] (q3) [left = 1.2cm of q4] {}; 
	  \node[p, minimum size=4ex] (q4) [below right = of q2] {};
      \node[a, minimum size=4.5ex] (q41) [below of = q21] {};  
 	   \node[fa, minimum size=4ex] (q6) [below of = q3] {};
	   \node[a, minimum size=4ex] (q7) [below of = q4] {};  
	   \node[p, minimum size=4ex] (q8) [below left = of q7] {}; 
	   \node[fp, minimum size=4ex] (q9) [right = 1.15cm of q8] {};
	   \node[fa, minimum size=4ex] (q10) [right = 1cm of q9] {};
	   \node[fp, minimum size=4ex] (q11) [right = 1cm of q10] {};
\begin{scope}[every node/.style={scale=1}]
  \path (q1) edge  node [pos=0.2] {} (q2);
  \path (q2) edge  node [pos=0.2,left] {} (q21);
  \path (q2) edge [black!30,thick] node [pos=0.25,right] {} (q3);
  \path (q2) edge [black!30,thick] node [pos=0.2,right] {} (q4);
  \path (q21) edge node [pos=0.7] {} (q41);
  \path (q4) edge  node  [pos=0.2, left] {} (q5);
  \path (q21) edge  [black!30] node  [pos=0.1, right] {} (q6);
  \path (q41) [bend left] edge node [pos=0.1, left] {} (q21);
  \path (q6) [black!30,bend right] edge node [pos=0.2,right] {} (q3);
  \path (q3) edge [black!30]  node [pos=0.5,left] {} (q6);
  \path (q7) edge node [pos=0.1,left] {} (q8);
  \path (q7) edge [black!30] node [pos=0.5,left] {} (q9);
  \path (q7) edge [black!30] node [pos=0.5,right] {} (q11);
  \path (q11) edge [black!30,bend right] node [pos=0.5,right] {} (q7);
  \path (q9) edge [black!30]  node [pos=0.5,below] {} (q10);
  \path (q10) edge [black!30]  node [pos=0.5,below] {} (q11);
   \path (q8) edge [bend left = 50] node [pos=0.1,left] {} (q7);

 \end{scope}
\end{tikzpicture}

}
\end{center}
\vspace{-0.2cm}
\caption{(A) An example of a game arena with a safety winning condition. (B) The nominal adviser $\alpha^0$ and $\losing$ via marking the forbidden transitions and states in $\losing$ in red. $\varsigma_p^{\alpha^0\star}$ is in green. (C) The non-blocking adviser restricted arena $\T^{\alpha^0}$. (D) $\alpha^\star$ and (E) The non-blocking adviser restricted arena $\T^{\alpha^\star}$.}
\label{fig:example3}
\vspace{-0.7cm}
\end{figure}
The guided system execution proceeds as follows: The system starts in state $s_\curr = s_{p_1}$ with $\alpha_\curr = \alpha^\star$ and $\varsigma_{p,\curr} = \varsigma_p^{\alpha^\star \star}$. Input $u_{p_1}$ is applied, $s_\curr = s_{2}$. Then, $\alpha_\curr(s_\curr) = \alpha^\star(s_2)$ is provided. The adversary chooses either $u_{a_1}, u_{a_2}$, or $u_{a_3}$, but, through the adviser it is recommended not to select $u_{a_3}$ (soft advise) and $u_{a_2}$ (hard advise). If the choice is $u_{a_1}$, the system state is updated to $s_\curr = s_3$, and in the remainder of the execution, the protagonist and the adversary apply $u_{p_2}$ and $u_{a_4}$, respectively, switching between states $s_3$ and $s_6$. If the choice is $u_{a_3}$, a soft advice is disobeyed, the current state becomes $s_5$ and the current adviser and strategy are updated to $\alpha_\curr = \alpha^0$ and $\varsigma_{p,\curr} = \varsigma_p^{\alpha^0 \star}$, which satisfy that $\lambda(\alpha^0) = \inf_{\alpha \in \A_\preceq} \lambda(\alpha)$. Input $u_{p_5}$ is then applied and $s_\curr = s_8$. In the remainder of the execution, the adversary is guided to follow the hard advices $u_{a_7},u_{a_8} \in \alpha_{\curr}(s_8)$, leading the system to switching between $s_8$ and $s_9$. If the choice in $s_2$ is $u_{a_2}$ despite the hard advice, the system reaches an unsafe state.
}

\end{example}


\section{Conclusions and future work}
We have studied the problem of synthesizing least-limiting guidelines for decision making in semi-autonomous systems involving entities that are uncontrollable, but partially willing to collaborate on achieving safety of the overall system. We have proposed a rigorous formulation of such problem and an algorithm to synthesize least-limiting advisers for an adversary in  a 2-player safety game and we have proposed a systematic way to guide the system execution with their use. 
As far as we are concerned, this paper presents one of the first steps towards studying the problem of synthesizing guidelines for uncontrollable entities. Future work naturally includes extensions to more complex winning conditions, different measures of level of violation, and continuous state spaces. We also plan to implement the algorithms and show their potential in a case study.
\label{sec:conc}


%
\bibliographystyle{abbrv}
\bibliography{refer}  

\begin{thebibliography}{10}

\bibitem{counter}
R.~Alur, S.~Moarref, and U.~Topcu.
\newblock Counter-strategy guided refinement of {GR(1)} temporal logic
  specifications.
\newblock In {\em Formal Methods in Computer-Aided Design}, pages 26--33.
  {IEEE}, 2013.

\bibitem{games-book}
K.~R. Apt and E.~{Gr\"adel}, editors.
\newblock {\em Lectures in Game Theory for Computer Scientists}.
\newblock Cambridge University Press, 2011.

\bibitem{model-repair}
E.~Bartocci, R.~Grosu, P.~Katsaros, C.~Ramakrishnan, and S.~Smolka.
\newblock Model repair for probabilistic systems.
\newblock In {\em Tools and Algorithms for the Construction and Analysis of
  Systems}, volume 6605 of {\em LNCS}, pages 326--340. Springer Berlin
  Heidelberg, 2011.

\bibitem{permissive}
J.~D. W.~I. Bernet, Julien.
\newblock Permissive strategies : from parity games to safety games.
\newblock {\em Theoretical Informatics and Applications}, 36(3):261--275, 2002.

\bibitem{brim}
L.~Brim, J.~Chaloupka, L.~Doyen, R.~Gentilini, and J.~Raskin.
\newblock Faster algorithms for mean-payoff games.
\newblock {\em Formal Methods in System Design}, 38(2):97--118, 2011.

\bibitem{environment}
K.~Chatterjee, T.~Henzinger, and B.~Jobstmann.
\newblock Environment assumptions for synthesis.
\newblock In {\em Concurrency Theory (CONCUR)}, volume 5201 of {\em LNCS},
  pages 147--161. Springer Berlin Heidelberg, 2008.

\bibitem{mr-mdp}
T.~Chen, E.~M. Hahn, T.~Han, M.~Kwiatkowska, H.~Qu, and L.~Zhang.
\newblock Model repair for {Markov} decision processes.
\newblock In {\em International Symposium on Theoretical Aspects of Software
  Engineering (TASE)}, pages 85--92. IEEE, 2013.

\bibitem{positional}
A.~Ehrenfeucht and J.~Mycielski.
\newblock Positional strategies for mean payoff games.
\newblock {\em International Journal of Game Theory}, 8(2):109--113, 1979.

\bibitem{faella}
M.~Faella.
\newblock Best-effort strategies for losing states.
\newblock {\em CoRR}, abs/0811.1664, 2008.

\bibitem{sandra}
S.~Hirche and M.~Buss.
\newblock Human-oriented control for haptic teleoperation.
\newblock {\em Proceedings of the IEEE}, 100(3):623--647, 2012.

\bibitem{marius-nondet}
M.~Kloetzer and C.~Belta.
\newblock Dealing with nondeterminism in symbolic control.
\newblock In {\em Hybrid Systems: Computation and Control (HSCC)}, pages
  287--300, Berlin, Heidelberg, 2008. Springer-Verlag.

\bibitem{hadas09TL}
H.~Kress-Gazit, G.~E. Fainekos, and G.~J. Pappas.
\newblock Temporal logic-based reactive mission and motion planning.
\newblock {\em IEEE Transactions on Robotics}, 25(6):1370--1381, 2009.

\bibitem{mining}
W.~Li, L.~Dworkin, and S.~A. Seshia.
\newblock Mining assumptions for synthesis.
\newblock In {\em ACM/IEEE International Conference on Formal Methods and
  Models for Codesign (MEMOCODE)}, 2011.

\bibitem{mazo}
M.~Mazo and M.~Cao.
\newblock Design of reward structures for sequential decision-making processes
  using symbolic analysis.
\newblock In {\em American Control Conference (ACC), 2013}, pages 4393--4398,
  2013.

\bibitem{ketan}
K.~Savla, T.~Temple, and E.~Frazzoli.
\newblock Human-in-the-loop vehicle routing policies for dynamic environments.
\newblock In {\em IEEE Conference on Decision and Control (CDC)}, pages
  1145--1150, 2008.

\bibitem{hscc}
J.~Tumova, G.~C. Hall, S.~Karaman, E.~Frazzoli, and D.~Rus.
\newblock Least-violating control strategy synthesis with safety rules.
\newblock In {\em Hybrid Systems: Computation and Control (HSCC)}, pages 1--10.
  ACM, 2013.

\bibitem{mp-sofsem}
A.~van Hulst, M.~Reniers, and W.~Fokkink.
\newblock Maximally permissive controlled system synthesis for modal logic.
\newblock In {\em Theory and Practice of Computer Science (SOFSEM)}, volume
  8939 of {\em LNCS}, pages 230--241. Springer Berlin Heidelberg, 2015.

\bibitem{nok-hscc2010}
T.~Wongpiromsarn, U.~Topcu, and R.~M. Murray.
\newblock {R}eceding horizon control for temporal logic specifications.
\newblock In {\em Hybrid Systems: Computation and Control (HSCC)}, pages
  101--110, 2010.

\end{thebibliography}
%
%


\end{document}